\documentclass[preprint,12pt]{elsarticle}

\usepackage{amssymb}
\usepackage{amsmath,amsfonts,amssymb,url}
\usepackage{caption}
\usepackage{subcaption}
\usepackage{algorithmic}
\usepackage{algorithm}
\usepackage{array}
\usepackage{textcomp}
\usepackage{stfloats}
\usepackage{url}
\usepackage{verbatim}
\usepackage{graphicx}
\usepackage{textcomp}
\usepackage{colortbl}
\usepackage{cite}
\usepackage{mathrsfs}
\usepackage{float}
\usepackage{geometry,textgreek,cjhebrew}
\newtheorem{prop}{Proposition}
\newtheorem{proof}{Proof}

\journal{Journal of The Franklin Institute}

\begin{document}

\begin{frontmatter}



\title{Goal-oriented compression for $L_p$-norm-type goal functions: Application to power consumption scheduling}

 \author[label1]{Yifei Sun}
 \author[label2]{Hang Zou}
 \author[label3]{Chao Zhang}
 \author[label4,label5]{Samson Lasaulce}
 \author[label1]{Michel Kieffer}
 \affiliation[label1]{organization={Universite Paris-Saclay - CNRS - CentraleSupelec - L2S},
             city={Gif-sur-Yvette},
             postcode={F-91192},
            country={France}}

\affiliation[label2]{organization={Technology Innovation Institute},
            city={Abu Dhabi},
            country={UAE}}

\affiliation[label3]{organization={School of Computer Science and Engineering, Central South University},
            city={Changsha},
            postcode={410083}, 
            country={China}}
\affiliation[label4]{organization={Khalifa University},
            city={Abu Dhabi},
            postcode={127788}, 
            country={UAE}}

\affiliation[label5]{organization={CRAN, CNRS-Universite de Lorraine},
            city={Nancy},
            postcode={F-54000}, 
            country={France}}

\begin{abstract}
Conventional data compression schemes aim at implementing a trade-off between the rate required to represent the compressed data and the resulting distortion between the original and reconstructed data. However, in more and more applications, what is desired is not reconstruction accuracy but the quality of the realization of a certain task by the receiver. In this paper, the receiver task is modeled by an optimization problem whose parameters have to be compressed by the transmitter. Motivated by applications such as the smart grid, this paper focuses on a goal function which is of $L_p$-norm-type. The aim is to design the precoding, quantization, and decoding stages such that the maximum of the goal function obtained with the compressed version of the parameters is as close as possible to the maximum obtained without compression. The numerical analysis, based on real smart grid signals, clearly shows the benefits of the proposed approach compared to the conventional distortion-based compression paradigm. 

\end{abstract}


\begin{highlights}
\item General framework for designing compression methods for the $L_p$ norm minimization problem.
\item Novel linear and nonlinear transformation schemes by taking into account the performance degradation in terms of the $L_p$ norm induced by model reduction. 
\item Tailor the quantization rule to be goal-oriented by considering the impact of the precoding and the final use of the compressed data.
\item Evaluation of the proposed coding schemes with a real dataset and show the significant performance improvement compared to existing conventional transformation and quantization techniques.
\end{highlights}

\begin{keyword}
Data compression, Goal oriented communications, Quantization, Learning, Neural networks, Precoding.


\end{keyword}

\end{frontmatter}


\section{Introduction}
With the development of new paradigms such as the industrial internet, the internet of things (IoT), the smart grid, or networked controlled systems for instance, networks become more and more distributed. Information exchanges between the different devices are necessary to implement cooperation or coordination and achieve a given goal. A huge amount of data is often generated and transmitted, such as in the smart grid \citep{compressionsmartgridstcheou2014,smartgridgeneratingaiello2014}. Because of practical limitations in terms of communication and computational resources, it is important if not necessary to compress the exchanged data.

There exists a quite solid literature on the problem of data compression in the smart grid. Compared to lossless compression (see e.g., \citep{losslesskraus2012}), lossy compression achieves a higher compression ratio at the expense of degrading the accuracy of data. Reference \citep{gerek2008compression} transforms power quality event data to 2D, showing the correlation between far away sample, which help to exploit the redundancy of the data. \citep{efficientcompressionzhang2011} compresses the estimated fundamental component and transient component by different compression techniques. A two-step compression is proposed in \citep{adaptivemultivariatechowdhury2021}. The dimension of the data is firstly reduced by using principal component analysis. The preserved components are further compressed by using compressed sampling. Due to the limit in terms of rate, quantization schemes \citep{quantizationgray1998} are often used to minimize distortion \citep{lossyberger1998}. 

{\color{black} Most of the aforementioned lossy compression techniques try to minimize the mean square error (MSE) while compressing the data. Alternative distortion measures have been considered for a long time, see \citep{quantizationgray1998} and the references therein. Indeed, it has been known that the MSE may be not suited for assessing the performance of various image processing operations (\emph{e.g.}, image segmentation or pattern recognition)}. Over the past years, the signal processing and digital communications communities have realized that the distortion-based data compression paradigm should be revisited \citep{zhang2017payoff,shlezinger2019hardware,popovski2020semantic,whatsemanticlan2021,semanticspkalfa2021,6Gstrinati2021, gunduz2022beyond} by adapting the compression scheme to the final use, task, or goal pursued by the receiver. The attempt to incorporate aspects such as semantics and effectiveness into communication theory is in fact not new. For instance, reference \citep{bar1953semantic} adopts a probabilistic logic approach to study the semantic aspect whereas \citep{goldreich2012theory} develops a complexity theory approach of the goal-oriented communication problem. Concerning the signal processing and communication point of view, which is the one of interest for the present paper, the literature of goal-oriented communications is still in full development. For instance, a goal-oriented signal processing technique, based on a graph-based semantic language and a goal filtering method, is proposed and applied to specific goals in \citep{semanticspkalfa2021}. The goal-oriented/semantic communication concept has also been developed to design the next cellular communication generation (see e.g., \citep{6Gstrinati2021, kountouris2021semantics, chaccour2022less, getu2023making}).


Concerning the goal-oriented approach for the data compression problem, the literature is equally relatively small but contains some relevant works that the present work can be related to. For data compression also, the signal processing community, has also been aware well before the recent publications on the topic that large compression gains might be reaped when considering the final use of the signal or image (see, \emph{e.g.}, \citep{roucos1983segment,soong1989phonetically,hirata1989loobit,lopes200340} for the case of speech compression). The goal-oriented quantization problem is posed and formalized for the first time in \citep{zhang2017payoff} and is tackled in a deeper manner in \citep{zhang2021decision}, \citep{zou2018decision}, \citep{shlezinger2019hardware}, \citep{shlezinger2021deep}, and in \citep{zou2023goal}. In \citep{mostaani2022task}, the problem of goal-oriented quantization is studied for the control problem of rendez-vous in a multi-agent setting. As for the goal-oriented precoding or data preprocessing problem, it has been introduced independently in  \citep{shlezinger2019hardware} and \citep{sun2019new}. 

Compared to the most related works, the present paper provides significant progresses into the design of goal-oriented data source encoders. \textcolor{black}{The paper presents a novel task-oriented compression scheme when the task can be modeled by an optimization problem, which is a very relevant model for communications and energy systems. Our approach precisely assumes the task can be represented by a function (to be maximized) whose variables are the decisions to be taken and whose parameters have to be compressed. The focus of this paper is on the $L_p-$norm minimization problem, which is an important problem for example for smart grid applications, and in particular, for the fundamental problem of power consumption scheduling (\emph{e.g.}, to minimize the total consumed energy price, Joule losses, or the peak power). More specifically,} for the goal-oriented precoding or transform stage, the best linear solution is developed and compared to a non-linear solution which is based on a convolution auto-encoder. 
In contrast with \citep{shlezinger2019hardware}, the proposed precoding scheme exploits the structure of the decision function (an approximation of it to be precise). 
Also, a goal-oriented quantizer is used not to quantize the input signal (as in \citep{mostaani2022task, zou2023goal}) but the goal-oriented precoder output. Additionally, the used quantizer is not assumed to be a set of uniform scalar quantizers as in \citep{shlezinger2019hardware} and also it is tailored to the goal, which is an $L_p$-norm-type utility function. We also investigate the problem of knowing how the two compression stages interact each other, to understand to what extent it is possible to accumulate the gains of the goal-oriented paradigm when applied to the two stages. 
Last but not least, the schemes are applied to real smart grid measurements, which leads to a detailed numerical performance analysis and discussing the design of goal-oriented precoding and quantization. \textcolor{black}{The paper does not only provide implementable coding schemes but also provides quantitative elements behind the intuition that accounting for the impact of compression noise on the task is beneficial. Our approach allows one to provide e.g., an analytical characterization of an approximation of the optimal precoder, a goal-oriented quantization algorithm which works for an arbitrary utility function and not only for the $L_p$-norm. A purely data-driven approach (based on neural networks) would not provide these structural elements which can be exploited both for interpretations and making implementation easier.}

The rest of the paper is organized as follows. We formulate the problem in Section \ref{sec:Problem-formulation}. In Section \ref{sec:linearapproximation} a linear approximation (LT) is considered for the optimal decision function. A linear and nonlinear transformation (NLT) are proposed in Section \ref{sec:lineartransformation} and Section \ref{sec:nonlinear} respectively. We  describe the proposed goal-oriented quantizer in Section \ref{sec:goq}. An iterative approach optimizing the linear transformation stage and the goal-oriented quantization stage is proposed in Section \ref{sec:iteratively-improve}. Section \ref{sec:numerical} provides the numerical performance analysis. The paper is concluded in Section \ref{sec:conclusion}.

\section{Problem formulation}

\label{sec:Problem-formulation}

Consider the following utility function 
\begin{equation}
u\left(x;\ell\right)=-||x+\ell||_{p}\label{eq:Utility}
\end{equation}
where $x=\left[x_{1},x_{2},\dots,x_{N}\right]^{\mathrm{T}}\in\mathbb{R}_{+}^{N}$
represent the vector of decision variables, $\ell=\left[\ell_{1},\ell_{2},\dots,\ell_{N}\right]^{\mathrm{T}}\in\mathbb{R}_{+}^{N}$
is a vector of parameters, and $\|\cdot\|_{p}$ is the $L_{p}$-norm,
\emph{i.e.}, for any $v\in\mathbb{R}^{n}$, $\|v\|_{p}=\left(|v_{1}|^{p}+\dots+|v_{n}|^{p}\right){}^{1/p}$
with $p\geqslant1$. Assuming that the sum of the decision variables
is lower-bounded as
\begin{equation}
\sum_{j=1}^{N}x_{j}\geqslant E\label{eq:Constraintx}
\end{equation}
where $E>0$ is a constant and that \eqref{eq:Utility} has to be
maximized, one obtains the following optimization problem
\begin{equation}
\begin{aligned}\underset{x}{\mathrm{maximize}} & \ u\left(x;\ell\right)\\
\text{s.t.}\  & \sum_{j=1}^{N}x_{j}-E\geqslant0\\
 & x_{j}\geqslant0,\ j=1,\dots,N.
\end{aligned}
\label{eq:xstar}
\end{equation}
The solution of \eqref{eq:xstar} is denoted
as $x^{\star}\left(\ell\right)$. The optimization problem \eqref{eq:xstar}
can model typical resource allocation problems in several applications. For instance, in the power consumption scheduling problem of energy
systems, $\ell$ represents the non-controllable part of the energy consumption of a given household over $N$ time slots
(\emph{e.g.}, consumed by the lighting and cooking appliances, TV, computers), $x$ represents the controllable
part (\emph{e.g.}, the desired state-of-charge of an electric vehicle (EV)) to be allocated over 
the $N$ time slots by some decision-making entity (scheduler), the constraint
\eqref{eq:Constraintx} corresponds to the minimum amount of energy that has to be provided. The utility function thus corresponds to minus the $L_{p}$-norm of the total consumption vector. When $p$ becomes large, maximizing the considering utility amounts to minimizing the peak power. When $p=2$ and $\ell$ and $x$ are interpreted as currents instead of powers, then it corresponds to the Joule losses minimization problem. 

We assume that only an approximated version $\widehat{\ell}\in\mathbb{R}^{N}$ of the non-controllable load vector $\ell$ is available to the scheduler which takes the decision $x$. This assumption can be not only motivated by the existence of limitations in terms of communication or computational resources but also for a need in terms of privacy. A lossy compression technique
is thus implemented to remove some redundancy from the source signal $\ell$. Considering
$\widehat{\ell}$ in place of $\ell$ in \eqref{eq:xstar}, the resulting
solution becomes $x^{\star}\left(\widehat{\ell}\right)$.

 \textcolor{black}{
\begin{figure}[ht]
\centering\includegraphics[scale=1.1]{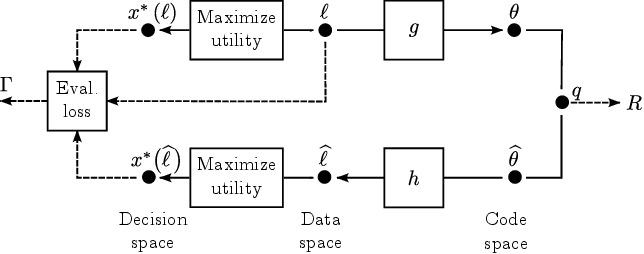} \caption{\textcolor{black}{Coding scheme targeting optimality loss minimization}}
\label{fig:Coding-scheme}
\end{figure}}


Instead of minimizing the reconstruction error as in conventional
compression schemes, this paper aims to design a compression scheme
for $\ell$ to mitigate the performance degradation resulting from
the distortion (compression noise) introduced on $\widehat{\ell}$,
\emph{i.e.}, to minimize the expected optimality loss 
\begin{equation}
\Gamma = \mathbb{E}_{\ell}\left[\left|u\left(x^{\star}\left(\ell\right);\ell\right)-u\left(x^{\star}\left(\widehat{\ell}\right);\ell\right)\right|^{2}\right],\label{eq:expoploss}
\end{equation}
where the expectation is over $\ell$. To define properly the arguments of $\Gamma$, it is first needed to define the different functions that need to be optimized. These functions appear in Figure~\ref{fig:Coding-scheme}. In this
setting, the compression noise stems both from precoding/decoding
and quantization. The precoding mapping 
\begin{equation}\label{eq:def-of-g}
\begin{array}{cccc}
g: & \mathbb{R}_{+}^{N} & \rightarrow & \mathbb{R}^{K}\\
 & \ell & \mapsto & \theta
\end{array}
\end{equation}
transforms the parameter vector $\ell$ to a vector $\theta$ of dimension
$K\leqslant N$, and the dimensionality reduction induces some information
loss. A quantizer 
\begin{equation}
\begin{array}{cccc}
q: & \mathbb{R}^{K} & \rightarrow & \mathbb{R}^{K}\\
 & \theta & \mapsto & \widehat{\theta}
\end{array}
\end{equation}
follows the precoder, which constitutes a second source of the compression
noise. The quantized vector $\widehat{\theta}$ is transmitted and
assumed received without error by the decision-making entity. A decoding
function 
\begin{equation}
\begin{array}{cccc}
h: & \mathbb{R}^{K} & \rightarrow & \mathbb{R}^{N}\\
 & \widehat{\theta} & \mapsto & \widehat{\ell}
\end{array}
\end{equation}
is then used to obtain an estimate $\widehat{\ell}$ of $\ell$.

The goal is to design the precoding, quantization, and decoding functions
that minimize the expected optimality loss \eqref{eq:expoploss}
\begin{equation}
\left(g^{\star},q^{\star},h^{\star}\right)\in\arg\min_{ \left(g,q,h\right)} \Gamma(g,q,h) =\mathbb{E}_{\ell}\left[\left|u\left(x^{\star}\left(\ell\right);\ell\right)-u(x^{\star}(\widehat{\ell});\ell)\right|^{2}\right]\label{eq:objective}
\end{equation}
where $\widehat{\ell}=h(q(g(\ell)))$.

Solving \eqref{eq:objective} directly is generally a hard task computationally speaking. This is the main motivation for searching for suboptimal solutions whose determination involve affordable complexity. In this paper, the functions $g$, $q$, and $h$ are therefore optimized separately.
More precisely, the precoding stage without quantization
noise is designed. In contrast with conventional linear transformation such as the Karhunen-Lo\`eve Transform (KLT), the precoding stage aims at minimizing the optimality loss in terms of the utility function $u$. The KLT is known to be the best linear transformed in terms of MSE but this optimality no longer holds for the optimality loss. Second,  the precoding/decoding
scheme is fixed and a goal-oriented quantization scheme is proposed to mitigate
the performance degradation brought by quantization noise. At last, the precoding/decoding stage and the quantization stage are optimized in an iterative manner. 


\section{Linear approximation of the optimal decision}

\label{sec:linearapproximation}

The impact of quantization on the optimality loss, as defined by (\ref{eq:expoploss}), can be seen to depend on the utility function $u$ and the optimal decision function $x^\star$. In \citep{zou2023goal}, it has been proved formally (in the high-resolution regime) how the regularity and smoothness properties of these functions impact the optimality loss. To be able to exploit the optimal decision function in the design of the precoding stage, we resort to a linear approximation of the former. In this section, we thus first characterize
the solution of the optimization problem \eqref{eq:xstar}. Then,
a linear approximation of this solution is provided to make the problem tractable.


\subsection{Optimal decision function}

For a given value of $\ell$, the solution of \eqref{eq:xstar}
is provided by Proposition~\ref{prop:opdecision} and involves a water-filling
approach, \textcolor{black}{which is widely used \citep{EVchargingBeaude2015,waterfillingshinwari2012,distributedwei2002}}.
\begin{prop}
\label{prop:opdecision} Consider a value of the parameter vector
$\ell$, and assume, without loss of generality, that 
\[
\ell_{1}\leqslant\dots\leqslant\ell_{j}\leqslant\dots\leqslant\ell_{N}.
\]
Then the components of $x^{\star}\left(\ell\right)$ are obtained
by
\begin{equation}
x_{j}^{\star}=\left(\mu-\ell_{j}\right)^{+},\label{eq:waterload}
\end{equation}
where 
\[
\mu=\frac{1}{n^{\star}}\left(E+\sum_{j=1}^{n^{\star}}\ell_{j}\right)
\]
indicates the water level, $\left(\cdot\right)^{+}=\max\left(\cdot,0\right)$,
and $n^{\star}$ is the number of non-zero components of $x^{\star}\left(\ell\right)$
evaluated as 
\begin{equation}
n^{\star}=\max_{n\leqslant N}\left\{ n:\left(n-1\right)\ell_{n}-\sum_{j=1}^{n-1}\ell_{j}\leqslant E\right\} .\label{eq:nstar}
\end{equation}
\end{prop}
\begin{proof} See~\ref{subsec:Proof-opdecison}. \end{proof}

\subsection{Linear approximation}

The solution provided by Proposition~\ref{prop:opdecision} has a
non-linear dependency in $\ell$, and this dependency is not explicit in general. To circumvent this difficulty we resort to a first order approximation of the optimal decision function. For this purpose, we evaluate the sensitivity of $x^{\star}\left(\ell\right)$
with respect to $\ell$, by considering the first-order Taylor expansion
of $x^{\star}$ around $\ell$ 
\begin{equation}
\begin{aligned}x^{\star}\left(\ell+d\ell\right) & =x^{\star}\left(\ell\right)+\boldsymbol{H}(\ell)d\ell+o\left(d\ell\right)\\
\end{aligned}
\end{equation}
where 
\[
\boldsymbol{H}(\ell)=\frac{\partial x^{\star}}{\partial\ell^{\mathrm{T}}}\left(\ell\right)
\]
is the Jacobian matrix of $x^{\star}\left(\ell\right)$ obtained using
Proposition~\ref{prop:linearapp}.
\begin{prop}
\label{prop:linearapp} Consider a value of the parameter vector $\ell$
such that $n^{\star}$ given by \eqref{eq:nstar} remains constant
over some neighborhood $\mathcal{N}\left(\ell\right)$ of $\ell$.
Without loss of generality, assume that 
\[
\ell_{1}\leqslant\dots\leqslant\ell_{n^{\star}}\leqslant\ell_{n^{\star}+1}\leqslant\dots\leqslant\ell_{N}.
\]
Then 
\begin{equation}
x^{\star}\left(\ell\right)= \boldsymbol{H}\left(\ell\right)\ell+b\left(\ell\right)
\label{eq:LinearApprox}
\end{equation}
with
\[
\begin{aligned} & \boldsymbol{H}\left(\ell\right)=\\
 & \left(\begin{array}{ccccccc}
-1+\frac{1}{n^{\star}} & \frac{1}{n^{\star}} & \text{\ensuremath{\cdots}} & \frac{1}{n^{\star}} & 0 & \cdots & 0\\
\frac{1}{n^{\star}} & -1+\frac{1}{n^{\star}} & \ddots & \vdots & \vdots &  & \vdots\\
\vdots & \ddots & \ddots & \vdots\frac{1}{n^{\star}} & \vdots &  & \vdots\\
\frac{1}{n^{\star}} & \cdots & \frac{1}{n^{\star}} & -1+\frac{1}{n^{\star}} & 0 & \cdots & 0\\
0 & \cdots & \cdots & 0 & 0 & \cdots & 0\\
\vdots &  &  & \vdots & \vdots & \ddots & \vdots\\
0 & \cdots & \cdots & 0 & 0 & \cdots & 0
\end{array}\right)
\end{aligned}
\]
and \textcolor{black}{} 
\[
b\left(\ell\right)=\left(\frac{E}{n^{\star}},\dots,\frac{E}{n^{\star}},0,\dots,0\right)^{\mathrm{T}}.
\]
\end{prop}
\begin{proof} See~\ref{subsec:Proof-linearapp}. \end{proof}

\section{Linear goal-oriented precoding\label{sec:lineartransformation}}

In this section, we investigate the problem of the determination of the best linear precoding or transform to be applied to $\ell$, which is to approximate $\ell$ by a weighted sum of a given number of vectors of a certain basis; the number of vectors of this sum precisely corresponds to the number $K \in \{1,\dots,N\}$ defined by (\ref{eq:def-of-g}). It is known that, for a given number of basis vector, the KLT provides the best basis in the sense of the MSE $\mathbb{E}_{\ell}(\| \ell - \widehat{\ell} \|^2)$ (see, \emph{e.g.}, \citep{mallat1999wavelet}). Nevertheless, this optimality result is no longer true in the presence of an arbitrary goal and in particular for the optimality loss definition used in this paper. The motivation of this section is therefore to propose a linear precoding or transformation scheme which is matched to the $L_p$-norm function.  

In this section, to simplify the analysis, the effect of quantization noise is assumed to be negligible. Under these assumptions,
one has that $\theta=g\left(\ell\right)=\boldsymbol{B}\ell$, and $\widehat{\ell}=h\left(\theta\right)=\boldsymbol{B}^{\mathrm{T}}\theta$,
where $\boldsymbol{B}\in\mathbb{R}^{K\times N}$ is the precoding
matrix and $\boldsymbol{B}^{\mathrm{T}}$ the decoding matrix. Then the expected
optimality loss \eqref{eq:expoploss} becomes 
\begin{equation}
\Gamma\left(\boldsymbol{B}\right) = \mathbb{E}_{\ell}\left[\left|u\left(x^{\star}\left(\ell\right);\ell\right)-u(x^{\star}(\boldsymbol{B}^{\mathrm{T}}\boldsymbol{B}\ell);\ell)\right|^{2}\right].\label{eq:opt_loss_linear}
\end{equation}
Two notational remarks are in order at this point. First, notice that in the above, a small abuse of notation is employed for clarity. Since, the optimality loss function is only considered with respect to the precoding stage, the arguments $q$ and $h$ are removed from $\Gamma$, and since only linear precoding is considered, the precoding function $g$ is replaced with the matrix $\boldsymbol{B}$. Second, motivated by practical considerations, in which exact statistics are not available but one has only access to a set of measurements, the notation $\widehat{\Gamma}_T$ will be use to refer to the empirical version of the optimality loss function $\Gamma$. Assuming that a dataset $\mathcal{L}=\{\ell^{\left(1\right)},\ell^{\left(2\right)},\dots,\ell^{\left(T\right)}\}$ of $T$ realizations or samples of $\ell$ is available, the empirical optimality loss expresses as
\begin{equation}
\begin{aligned}\widehat{\Gamma}_T\left(\boldsymbol{B}\right)= & \frac{1}{T}\sum_{i=1}^{T}\left|u(x^{\star}(\ell^{(i)});\ell^{(i)})-u(x^{\star}(\boldsymbol{B}^{\mathrm{T}}\boldsymbol{B}\ell^{(i)});\ell^{(i)})\right|^{2}.\end{aligned}
\label{eq:GammaB}
\end{equation}
Using the linear form \eqref{eq:LinearApprox} introduced in Proposition~\ref{prop:linearapp},
one obtains
\begin{align}
\widehat{\Gamma}_T\left(\boldsymbol{B}\right) & =\frac{1}{T}\sum_{i=1}^{T}\left|u(\boldsymbol{H}\left(\ell^{(i)}\right)\ell^{(i)}+b\left(\ell^{(i)}\right);\ell^{(i)})\right.\nonumber \\
 & -\left.u\left(\boldsymbol{H}\left(\boldsymbol{B}^{\mathrm{T}}\boldsymbol{B}\ell^{(i)}\right)\boldsymbol{B}^{\mathrm{T}}\boldsymbol{B}\ell^{(i)}+b\left(\boldsymbol{B}^{\mathrm{T}}\boldsymbol{B}\ell^{(i)}\right);\ell^{(i)}\right)\right|^{2}.\label{eq:GammaB2}
\end{align}
To minimize $\widehat{\Gamma}_T$ we resort to a gradient descent algorithm. To simplify the computation procedure of the derivatives of $\widehat{\Gamma}_T$, it is assumed that a small variation of $\boldsymbol{B}$ does not change
the entries of $\boldsymbol{H}\left(\boldsymbol{B}^{\mathrm{T}}\boldsymbol{B}\ell^{(i)}\right)$
and $b\left(\boldsymbol{B}^{\mathrm{T}}\boldsymbol{B}\ell^{(i)}\right)$ for
all $i=1,\dots,T$. This assumption allows one to state the following proposition.
\begin{prop}
\label{prop:gradient} Consider $p\in\mathbb{N}^{+}$. The matrix
containing the derivatives of $\widehat{\Gamma}_T\left(\boldsymbol{B}\right)$
with respect to the components of $\boldsymbol{B}$ is
\begin{equation}
\nabla_{\boldsymbol{B}}\widehat{\Gamma}_T\left(\boldsymbol{B}\right)=\frac{1}{T}\sum_{i=1}^{T}C_{i}\left(\boldsymbol{B}\ell^{\left(i\right)}\beta_{i}^{\mathrm{T}}\boldsymbol{H}_{i}+\boldsymbol{B}\boldsymbol{H}_{i}^{\mathrm{T}}\beta_{i}\ell^{\left(i\right)T}\right)\label{eq:ybd}
\end{equation}
where 
\begin{equation}
C_{i}=2 \left(u\left(x^{\star}\left(\ell^{\left(i\right)}\right);\ell^{\left(i\right)}\right)-u\left(x^{\star}\left(\widehat{\ell}^{\left(i\right)}\right);\ell^{\left(i\right)}\right)\right)||x^{\star}\left(\widehat{\ell}^{\left(i\right)}\right)+\ell^{\left(i\right)}||_{p}^{1-p}
\end{equation}
\begin{equation}
\beta_{i}=\text{\ensuremath{\underbrace{\left(x^{\star}\left(\widehat{\ell}^{\left(i\right)}\right)+\ell^{\left(i\right)}\right)\odot\cdots\odot\left(x^{\star}\left(\widehat{\ell}^{\left(i\right)}\right)+\ell^{\left(i\right)}\right)}_{p-1}}}
\end{equation}
where $\odot$ indicates Hadamard product.

Moreover, when $p\rightarrow+\infty$, \eqref{eq:ybd} boils down
to 
\begin{equation}
\nabla_{\boldsymbol{B}}\widehat{\Gamma}_T\left(\boldsymbol{B}\right)=\frac{1}{T}\sum_{i=1}^{T}D_{i}\left(\boldsymbol{B}\ell^{\left(i\right)}s_{k\left(i\right)}^{\mathrm{T}}\boldsymbol{H}_{i}+\boldsymbol{B}\boldsymbol{H}_{i}^{\mathrm{T}}s_{k\left(i\right)}\ell^{\left(i\right)T}\right)\label{eq:ybdi}
\end{equation}
where 
\begin{equation}
D_{i}=2\left(u\left(x^{\star}\left(\ell^{\left(i\right)}\right);\ell^{\left(i\right)}\right)-u\left(x^{\star}\left(\widehat{\ell}^{\left(i\right)}\right);\ell^{\left(i\right)}\right)\right)
\end{equation}
\begin{equation}
s_{k\left(i\right)}=\left(\begin{array}{c}
0_{(k\left(i\right)-1)\times1}\\
1\\
0_{(N-k\left(i\right))\times1}
\end{array}\right)
\end{equation}
and
\[
k\left(i\right)=\arg\max_{k}x_{k}^{\star}\left(\widehat{\ell}^{\left(i\right)}\right)+\ell_{k}^{\left(i\right)}.
\]
\end{prop}
\begin{proof} See~\ref{subsec:Proof-gradient} \end{proof}

Using \eqref{eq:ybd} or \eqref{eq:ybdi}, a local minimization of
$\widehat{\Gamma}_T\left(\boldsymbol{B}\right)$ can be performed by using a gradient
descent algorithm (see Algorithm~\ref{alg1}) to obtain $\boldsymbol{B}^{\star}$.
The search is initialized with the $K$ first vectors of the KLT built from the empirical covariance matrix obtained by using the vectors of the dataset $\mathcal{L}$. 

\begin{algorithm}
\caption{Gradient descent search $\boldsymbol{B}^{*}$}
\label{alg1} \begin{algorithmic} 

\global\long\def\algorithmicrequire{\textbf{Input:}}%


\REQUIRE Initialize $\boldsymbol{B}$ with the $K$ first vectors of Karhunen-Lo\`eve transform

\REQUIRE Initialize $i=0$, 

\WHILE{$i<it_{\max}$ and optimality loss reduced more
than $0.01\%$}

\STATE Compute the gradient: $\boldsymbol{G}\leftarrow\nabla_{\boldsymbol{B}}\widehat{\Gamma}_T\left(\boldsymbol{B}\right)$ 

\STATE Perform line search to get $\lambda$ such that $\widehat{\Gamma}_T\left(\boldsymbol{B}-\lambda\boldsymbol{G}\right)<\widehat{\Gamma}_T\left(\boldsymbol{B}\right)$

\STATE Update: $\boldsymbol{B}=\boldsymbol{B}-\lambda\boldsymbol{G}$ 

\STATE $i\leftarrow i+1$
\ENDWHILE

\end{algorithmic} 
\end{algorithm}

\section{Nonlinear goal-oriented precoding\label{sec:nonlinear}}

The use of linear transforms for data compression is largely motivated by complexity issues. Nevertheless, linear transforms are generally not optimal. They are optimal in terms of MSE when the input signal corresponds to realizations of a (vector) Gaussian random variable \citep{mallat1999wavelet}. A natural question is thus to assess the benefits of a nonlinear transform in the presence of goal functions such as the $L_p$-norm. This is why we consider here a larger class of precoders for which the parameters can be learned. The aim of the precoder is to obtain a latent representation
$\theta\in\mathbb{R}^{K}$ of $\ell$ as follows $\theta=g\left(\ell,\Phi\right)$,
where $\Phi$ is a vector of parameters. Then, the decoder takes a
possibly quantized version of the latent representation as input to
obtain an estimate of $\widehat{\ell}$ as $\widehat{\ell}=h\left(\theta,\Psi\right)$,
where $\Psi$ is a vector of parameters.

Auto-encoders appear as a natural tool to implement a nonlinear transform when the purpose is to perform model reduction. Auto-encoders have been previously developed for image compression
\citep{balle_end--end_2017,balle_variational_2018,lee_context-adaptive_2019}
and are considered here to obtain a suited low-dimension representation of $\ell$. 



A convolutional auto-encoder consists of an encoding stage and a decoding stage
as illustrated in Figure~\ref{fig:StructureCNN}. The encoding stage
is simply a concatenation of convolution layers followed by a fully connected
layer in the end while the decoding stage implements the inverse structure of the encoding one. 
\begin{figure}[ht]
\centering\includegraphics[scale=0.5]{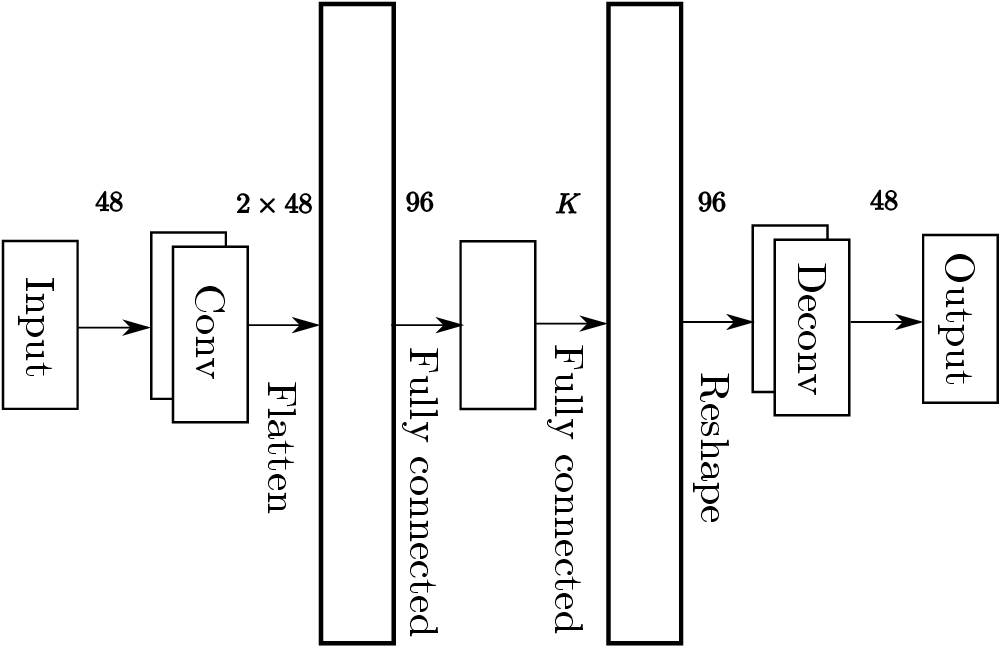}
\caption{{\color{black}Structure of the considered convolutional auto-encoder to evaluate the nonlinear goal-oriented precoding and decoding functions}}
\label{fig:StructureCNN} 
\end{figure}

The convolutional auto-encoder aims at searching a pair of $\left(\Phi^{\star},\Psi^{\star}\right)$ minimizing the expected optimality loss
\begin{equation}
\left(\Phi^{\star},\Psi^{\star}\right) \in \arg\min_{\left(\Phi,\Psi\right)}\mathbb{E}_{\ell}\left[\left|u\left(x^{\star}\left(\ell\right);\ell\right)-u(x^{\star}(h\left(g\left(\ell,\Phi\right),\Psi\right));\ell)\right|^{2}\right].
\end{equation}
As in the previous section, we also make a small notational abuse, since $\Psi$ and $\Phi$ are taken as the arguments of $\widehat{\Gamma}_T$.

Assume a set $\mathcal{L}$ of realizations of $\ell$ is available, $\left(\Phi^{\star},\Psi^{\star}\right)$ can be trained through the loss function defined as
\begin{equation}
\widehat{\Gamma}_T\left(\Phi,\Psi\right)  = \frac{1}{T}\sum_{i=1}^{T}\left\vert u(x^{\star}\left(\ell^{(i)}\right);\ell^{(i)})-u\left(x^{\star}\left(h\left(g\left(\ell^{(i)},\Phi\right),\Psi\right)\right);\ell^{(i)}\right)\right\vert ^{2}.\label{eq:Gammaphipsi}
\end{equation}

More precisely, as for the encoding stage in our simulation, one convolutional layer involving two kernels with size $5$ is considered. Zero padding is implemented to impose the output of each kernel having the same dimension as the input. The outputs of two kernels are reshaped to a vector of dimension $2N$. A fully connected layer between the $2N$ neurons and $K$ neurons is considered to get the latent representation $\theta$. The decoding part is inverse processing to reconstruct the data $\widehat{\ell}$ with dimension $N$.

\section{Goal-oriented quantization\label{sec:goq}}

In this section, we assume that the precoding and decoding functions
$g\left(\cdot\right)$ and $h\left(\cdot\right)$ are given. We focus
on the optimization of the quantizer design in order to minimize the
optimality loss for a given choice of $g\left(\cdot\right)$ and $h\left(\cdot\right)$.

\textcolor{black}{For the choice of quantization schemes, vector quantization is considered in this paper. Targeting signal reconstruction quality, using uniform scalar quantizers followed by entropy coding constitutes an overwhelmingly popular solution; it is adopted in neural compression schemes \citep{balle_end--end_2017}. Nevertheless, when the principal goal pursued is not to reconstruct the original data, it has been shown that element-wise uniform scalar quantization (namely, using a bank of scalar quantizers) could induce a significant performance degradation for the task to be executed (\citep{zhang2021decision,zou2023goal}). This can be explained by the fact that the different elements of the vector of variables to be compressed can have a markedly different influence on the final task; therefore, treating these elements equally is generally suboptimal. As a consequence, we use vector quantization in our scheme for task-oriented communication. Moreover, vector quantization allows one to ensure a fixed-rate quantized output, in contrast with auto-encoders accounting for rate constraints such as \citep{balle_end--end_2017}. At last, our choice is very well suited for the power consumption scheduling problem exploiting smart meter measurements, which have a relatively small dimension.}

A vector quantizer partitions the space $\mathbb{R}^{K}$ of the encoded
parameter $\theta=g\left(\ell\right)$ into several disjoint quantization
regions $\mathcal{C}_{1},\dots,\mathcal{C}_{M}$, \emph{i.e.}, 
\[
\bigcup_{i=1}^{M}\mathcal{C}_{i}=\mathbb{R}^{K}
\]
and $\mathcal{C}_{i}\bigcap\mathcal{C}_{j}=\varnothing$ for all $i\neq j$.
The quantization function is such that 
\begin{equation}
q\left(\theta\right)=r_{i}\iff\theta\in\mathcal{C}_{i},
\end{equation}
where $\mathcal{R}=\left\{ r_{1},\dots,r_{M}\right\} $ is the set
of representatives associated with the quantization regions $\mathcal{C}_{1},\dots,\mathcal{C}_{M}$
respectively. 
Contrary to conventional quantizers, which usually minimize the mean-square
reconstruction error, in what follows, we search for a pair $\left(\mathcal{R}^{\star},\mathcal{C}^{\star}\right)$
that minimizes the expected optimality loss 
\begin{align}
 & \mathbb{E}_{\ell}\left[\left|u\left(x^{\star}\left(\ell\right);\ell\right)-u\left(x^{\star}\left(\widehat{\ell}\right);\ell\right)\right|^{2}\right]=\nonumber \\
 & \sum_{m=1}^{M}\int_{\ell\in\mathcal{L}_{m}}\hspace{-2mm}\left|u\left(x^{\star}\left(\ell\right);\ell\right)-u\left(x^{\star}\left(h\left(r_{m}\right)\right);\ell\right)\right|^{2}\phi\left(\ell\right)\mathrm{d}\ell,\label{eq:GOQ}
\end{align}
where $\phi\left(\ell\right)$ is the probability density function
of $\ell$ and 
\[
\mathcal{L}_{m}\triangleq\{\ell\in\mathbb{R}^{N}|g\left(\ell\right)\in\mathcal{C}_{m}\},1\leqslant m\leqslant M.
\]
Finding jointly $\left(\mathcal{R}^{\star},\mathcal{C}^{\star}\right)$
is not trivial. Moreover, the evaluation of \eqref{eq:GOQ} considering
$\phi\left(\ell\right)$ is complex, even if $\phi\left(\ell\right)$
is perfectly known. Consequently, a practical algorithm is proposed
which is similar to the decisional quantizer proposed in \citep{zou2018decision}.
The main steps of this algorithm are detailed in what follows. As
in Section~\ref{sec:lineartransformation}, a set $\mathcal{L}=\{\ell^{\left(1\right)},\ell^{\left(2\right)},\dots,\ell^{\left(T\right)}\}$
of realizations of the parameter $\ell$ is used to approximate \eqref{eq:GOQ}.

Assume that at iteration $j$, a set of representatives $\mathcal{R}^{\left(j\right)}=\left\{ r_{1}^{\left(j\right)},\dots,r_{M}^{\left(j\right)}\right\} $
is available.
\begin{enumerate}
\item From $\mathcal{R}^{\left(j\right)}$, partition the set $\mathcal{L}$
as 
\begin{equation}
\mathcal{L}_{m}^{\left(j\right)}=\left\{ \ell\in\mathcal{L}\mid\mathcal{E}\left(r_{m}^{\left(j\right)};\ell\right)=\min_{i}\mathcal{E}\left(r_{i}^{\left(j\right)};\ell\right)\right\}, 1\leqslant m\leqslant M,\label{eq:Ljm}
\end{equation}
where
\begin{equation}
\mathcal{E}\left(r_{m};\ell\right)=\left|u\left(x^{\star}\left(\ell\right);\ell\right)-u\left(x^{\star}\left(h\left(r_{m}\right)\right);\ell\right)\right|^{2}.
\end{equation}
\item Update the set of representatives $\mathcal{R}^{\left(j+1\right)}=\left\{ r_{1}^{\left(j+1\right)},\dots,r_{M}^{\left(j+1\right)}\right\}$
as
\begin{equation}
r_{m}^{\left(j+1\right)}\in\arg\min_{r}\sum_{\ell\in\mathcal{L}_{m}^{\left(j\right)}}\mathcal{E}\left(r;\ell\right), 1\leqslant m\leqslant M.\label{eq:dstar}
\end{equation}
For any $\ell\in\mathbb{R}_{+}^{K}$, the goal-oriented quantizer
output is then obtained as
\[
q^{\left(j+1\right)}\left(g\left(\ell\right)\right)=r_{m}^{\left(j+1\right)}
\]
with
\[
m=\arg\min_{i=1,\dots,M}\mathcal{E}\left(r_{i}^{\left(j+1\right)};\ell\right).
\]
\item Evaluate the resulting estimate of the expected optimality loss as
\begin{equation}
\Gamma^{\left(j+1\right)}= \sum_{m=1}^{M}\sum_{\ell\in\mathcal{L}_{m}^{\left(k\right)}}\left|u\left(x^{\star}\left(\ell\right);\ell\right)
 -u\left(x^{\star}\left(h\left(q^{\left(j+1\right)}\left(g\left(\ell\right)\right)\right)\right);\ell\right)\right|^{2}.\label{eq:GOQ-1}
\end{equation}
\end{enumerate}
Like Linde-Buzo-Gray (LBG) algorithm \citep{algorithm1980linde},
Algorithm~\ref{alg2} performs this iterative evaluation to determine
a locally optimal goal-oriented quantizer. The set of representatives
$\mathcal{R}^{(0)}$ is initialized considering, $e.g.$, randomly
chosen elements of the set $g\left(\mathcal{L}\right)=\left\{ g\left(\ell^{\left(1\right)}\right),g\left(\ell^{\left(2\right)}\right),\dots,g\left(\ell^{\left(T\right)}\right)\right\} $.
\begin{algorithm}
\caption{Goal-oriented quantizer design algorithm}
\label{alg2} \begin{algorithmic} 

\REQUIRE Utility function $u\left(x;\ell\right)$ 

\global\long\def\algorithmicrequire{\textbf{Output:}}%
 \REQUIRE $\mathcal{R}^{\star}$ 

\global\long\def\algorithmicrequire{\textbf{Initialization:}}%
\REQUIRE $\mathcal{R}^{(0)}=\{r_{1}^{(0)},\dots,r_{M}^{(0)}\}$,
$j=1$, $\Gamma^{\left(0\right)}=\infty$

\WHILE{$j\leqslant j_{\max}$ and optimality loss uced
more than $0.01\%$} 

\STATE Evaluate $\mathcal{L}_{m}^{\left(j\right)},1\leqslant m\leqslant M,$
from $\mathcal{R}^{(j-1)}$ using \eqref{eq:Ljm}

\STATE Evaluate $r_{m}^{\left(j+1\right)},$$m=1,\dots,M$, using
\eqref{eq:dstar}

\STATE Evaluate $\Gamma^{\left(j+1\right)}$ using \eqref{eq:GOQ-1}

\STATE $j=j+1$

\ENDWHILE

\end{algorithmic} 
\end{algorithm}
As LBG algorithm, Algorithm~\ref{alg2} may only converge to a local
minimum of the expected optimality loss. The convergence depends on
the initialization of $\mathcal{R}^{(0)}$.

\section{Iterative optimization of linear transform and goal-oriented quantization}
\label{sec:iteratively-improve} 

Using Algorithm~\ref{alg1} and
Algorithm~\ref{alg2}, one obtains a pair of separately designed
precoder/decoder pair and quantizer. In Section~\ref{sec:lineartransformation},
Algorithm~\ref{alg1} neglects the impact of the quantization noise
in the design of the linear precoder/decoder. This section presents
an iterative design, accounting for the quantization noise in Algorithm~\ref{alg1},
to further uce the optimality loss.

The effect of quantization in Section~\ref{sec:goq} is to add some
noise $\eta$ to the output of the precoder 
\begin{align*}
\widehat{\theta} & =q\left(\theta\right)\\
 & =\theta+\eta.
\end{align*}
We assume that $\eta$ is Gaussian with mean $\mu_{\eta}\in\mathbb{R}^{K}$
and covariance matrix $\Sigma_{\eta}$. Considering the linear precoder/decoder
pair, the reconstructed signal becomes 
\begin{equation}
\begin{aligned}\widehat{\ell}= & \boldsymbol{B}^{\mathrm{T}}\widehat{\theta}\\
= & \boldsymbol{B}^{\mathrm{T}}\left(\theta+\eta\right)\\
= & \widetilde{\ell}+\boldsymbol{B}^{\mathrm{T}}\eta
\end{aligned}
\label{eq:quantizationmodel}
\end{equation}
where $\widetilde{\ell}=\boldsymbol{B}^{\mathrm{T}}\boldsymbol{B}\ell=\boldsymbol{B}^{\mathrm{T}}\theta$
is reconstructed from $\theta$ in absence of quantization.

Consider the $j$-th iteration of an iterative optimization algorithm
of the precoder/decoder pair and of the quantizer. Assume that a precoding
matrix $\boldsymbol{B}^{\left(j\right)}$, and a set of representatives
$\mathcal{R}^{\left(j\right)}$ have been obtained from Algorithms~\ref{alg1}
and~\ref{alg2}. The characteristics of the quantization noise can
be evaluated using the set $\mathcal{L}$ of realizations of $\ell$
as 
\begin{align*}
\widehat{\mu}^{\left(j\right)} & =\frac{1}{T}\sum_{i=1}^{T}\eta^{\left(i\right)}\\
\widehat{\Sigma}^{\left(j\right)} & =\frac{1}{T}\sum_{i=1}^{T}\left(\eta^{\left(i\right)}-\widehat{\mu}_{\eta}\right)\left(\eta^{\left(i\right)}-\widehat{\mu}_{\eta}\right)^{\mathrm{T}}
\end{align*}
\textcolor{black}{where 
\begin{equation}
\begin{aligned}\eta^{\left(i\right)}= & q\left(\boldsymbol{B}^{\left(j\right)}\ell^{\left(i\right)}\right)-\boldsymbol{B}^{\left(j\right)}\ell^{\left(i\right)}\end{aligned}
.
\end{equation}
}

\textcolor{black}{At iteration $j+1$, the impact of the quantization
noise can be introduced by introducing for each $\ell^{\left(i\right)}\in\mathcal{L}$,
$\overline{\kappa}$ noisy reconstructions of $\ell^{\left(i\right)}$
\begin{equation}
\widehat{\ell}^{\left(i,\kappa\right)}=\left(\boldsymbol{B}^{\left(j\right)}\right)^{\mathrm{T}}\boldsymbol{B}^{\left(j\right)}\ell^{\left(i\right)}+\left(\boldsymbol{B}^{\left(j\right)}\right)^{\mathrm{T}}\eta^{\left(i,\kappa\right)}\label{eq:NoisyFakeReconstruction}
\end{equation}
where $\eta^{\left(i,\kappa\right)}$, $1\leqslant\kappa\leqslant\overline{\kappa}$
are realizations of independent and identically distributed Gaussian
random vectors with mean $\widehat{\mu}^{\left(j\right)}$ and covariance
matrix $\widehat{\Sigma}^{\left(j\right)}$. The noisy reconstructions
$\widehat{\ell}^{\left(i,\kappa\right)}$ are then used in the evaluation
of the gradient }\eqref{eq:ybd} or \eqref{eq:ybdi} to obtain a precoder/decoder
optimization accounting for the quantization noise. As a result, a
matrix $\boldsymbol{B}^{\left(j+1\right)}$ is obtained. Algorithm~\ref{alg2}
is then applied to get $\mathcal{R}^{\left(j+1\right)}$. Then $\widehat{\mu}^{\left(j+1\right)}$
and $\widehat{\Sigma}^{\left(j+1\right)}$ can be evaluated. 

The iterative optimization process is summarized in Algorithm~\ref{algo:algorithm3}. \textcolor{black}{The computational complexity of linear precoding part is $O\left(TN^{2}K\right)$,  the complexity of the goal-oriented quantization part is $O\left(N_{rep}(TNK+C_r)\right)$, where $N_{rep}$ is the overall number of representatives and $C_r$ represents the complexity of computing the representatives of the $L_p$ norm optimization problem.}
\begin{algorithm}
\caption{Iterative optimization of the linear transformation and quantization}

\begin{algorithmic} 
\global\long\def\algorithmicrequire{\textbf{Input:}}%
\REQUIRE Initial matrix $\boldsymbol{B}^{\left(0\right)}$ (KL basis),
utility function $u\left(x;\ell\right)$, initial $\mathcal{R}^{\left(0\right)}$,
$\mu^{\left(0\right)}$, and $\Sigma^{\left(0\right)}$;\\

\global\long\def\algorithmicrequire{\textbf{Output:}}%
\REQUIRE $\boldsymbol{B}^{\star}$ and $\mathcal{R}^{\star}$;

\global\long\def\algorithmicrequire{\textbf{Initialization:}}%
\REQUIRE $j=1$;

\WHILE{$j\leqslant$$j_{\text{max}}$ and optimality loss reduced
more than $0.01\%$}

\STATE Evaluate $\boldsymbol{B}^{\left(j\right)}$ using Algorithm~\ref{alg1}
and \eqref{eq:NoisyFakeReconstruction}

\STATE Evaluate $\mathcal{R}^{\left(j\right)}$ using Algorithm \ref{alg2}

\STATE Evaluate $\mu^{\left(j\right)}$ and $\Sigma^{\left(j\right)}$

\STATE Evaluate $\Gamma^{\left(j\right)}$ using \eqref{eq:GOQ-1}

\STATE $j=j+1$

\ENDWHILE

\end{algorithmic} \label{algo:algorithm3} 
\end{algorithm}

\section{Numerical performance analysis}
\label{sec:numerical}

In this section, we conduct a comprehensive numerical analysis to provide insights into the benefits of goal-oriented compression. 
Our focus is on the power consumption scheduling problem in which the decision maker aims to find a controllable consumption vector $x$ minimizing the $L_p$ norm given a (perfect) day-ahead forecast of the non-controllable vector $\ell$ (the case of considering the forecast noise can be treated as an extension of this work). 
Due to the communication resource limitation, the decision maker has the sole knowledge of an approximated version of $\ell$, namely, $\widehat{\ell}$. Our goal is to mitigate the performance degradation induced by the deviation between $\widehat{\ell}$ and ${\ell}$. To assess the influence of the compression loss on the goal, we define the relative squared optimality loss (RSOL) of the compression scheme with respect to the ideal case as follows:
\begin{equation}
\rho_C(\%)=100\frac{\sum_{i=1}^{T}\left(U^{\left(i\right)}_{\mathrm{perfect}}-U^{\left(i\right)}_{C}\right)^{2}}{\sum_{i=1}^{T}\left(U^{\left(i\right)}_{\mathrm{perfect}}\right)^{2}},
\label{eq:ROLdef}
\end{equation}
 where $U^{\left(i\right)}$
 is obtained from the realization $\ell^{\left(i\right)}$, $C$ represents the compression scheme, and the performance of the ideal case is obtained by assuming that the controllable consumption $\ell$ is perfectly known by the receiver (that is, the decision-maker).

To evaluate our approach in practical scenarios, consumption profiles from the Ausgrid \citep{ausgrid} database are used for all the considered schemes. {\color{black} The data from Ausgrid consist of daily energy consumptions sampled every half hour for one year (from $01-07-2012$ to $30-06-2013$) for 300 users, thus the dimension of the data is $48\times(365\times300)$, that is $N=48$. Data have been randomly shuffled. Then $80~\%$ of the data have been used to train the methods and the remaining $20~\%$ for evaluation}. The energy need in terms of controllable consumption is set to $E=50$ kWh. As the system consists of several stages including precoding, quantization and decoding, we present the simulation results from different aspects for a comprehensive illustration.

\subsection{Precoding}
We first evaluate the benefit of using goal-oriented precoding schemes in the scenario without quantization noise. 
{\color{black} A KLT-based precoding scheme and an auto-encoder (AE) serve as references. The AE scheme is an adaptation to the compression for 1D vectors of the Tensorflow learned data compression network \citep{TensorflowCompTut}, itself based on the \citep{balle_end--end_2017}. The architecture consists of three parts: an encoder that transforms the input vector into a lower-dimensional latent variable;  a decoder that reconstructs the original source vector from this latent representation; a prior and entropy model between the encoder and decoder to model the marginal distribution of the latent variable and efficiently encode it to minimize the average code length. In this first part, the AE is trained to minimize the reconstruction mean-square error (MSE) as commonly used.}
When $K=1$, Figure~\ref{fig:comparewrtp} clearly demonstrates the improvement provided by the proposed goal-oriented transformation schemes over the conventional KLT and the AE. 
The nonlinear transformation scheme obtained from our neural network architecture outperforms the solution given by Algo.~1, since the activation function brings nonlinearity to the scheme. Moreover, one can observe that the RSOL increases with $p$. 
This can be explained by the fact that the denominator of RSOL defined by (\ref{eq:ROLdef}) usually decreases when $p$ increases. For the remaining simulations, we fix $p=\infty$, corresponding to the peak power minimization problem. 
\begin{figure}
	\centering
 \includegraphics[scale=0.7]{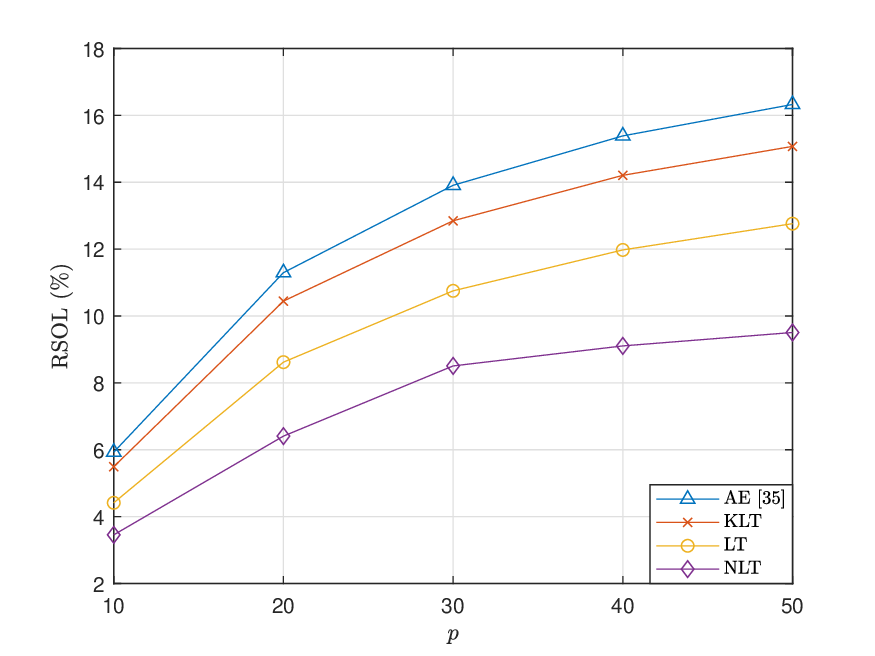}
 
    \caption{\textcolor{black}{Relative squared optimality loss ($\%$) v.s. the exponent of $L_{p}$-norm ($p$) when $K=1$. Both goal-oriented linear transform (LT) and goal-oriented non-linear transform (NLT) have a lower RSOL than the Karhuenen-Loeve transform (KLT) and the auto-encoder (AE transform). Among all schemes, the proposed NLT solution provides the best result. The influence of $p$ on the RSOL value shows also the influence of the goal function on the benefits of the goal-oriented approach.}}
	\label{fig:comparewrtp}
\end{figure}

{\color{black} To evaluate the influence of dimensionality reduction on the goal-oriented transformation schemes, we plot the RSOL w.r.t. the value of $K$.}
For the KLT, the orthogonal basis is obtained as eigenvectors of the covariance matrix, and the basis vectors are sorted from the most to the least important in terms of minimizing reconstruction error. \textcolor{black}{One can observe that the proposed NLT scheme outperforms other techniques. Moreover, while the AE structure leverages the nonlinear capabilities of neural networks, its lack of task-specific focus results in greater performance degradation when using a smaller $K$. This loss in performance can be substantially mitigated by employing a latent space with larger dimensions.}
\begin{figure}[H]
	\centering
 \includegraphics[scale=0.7]{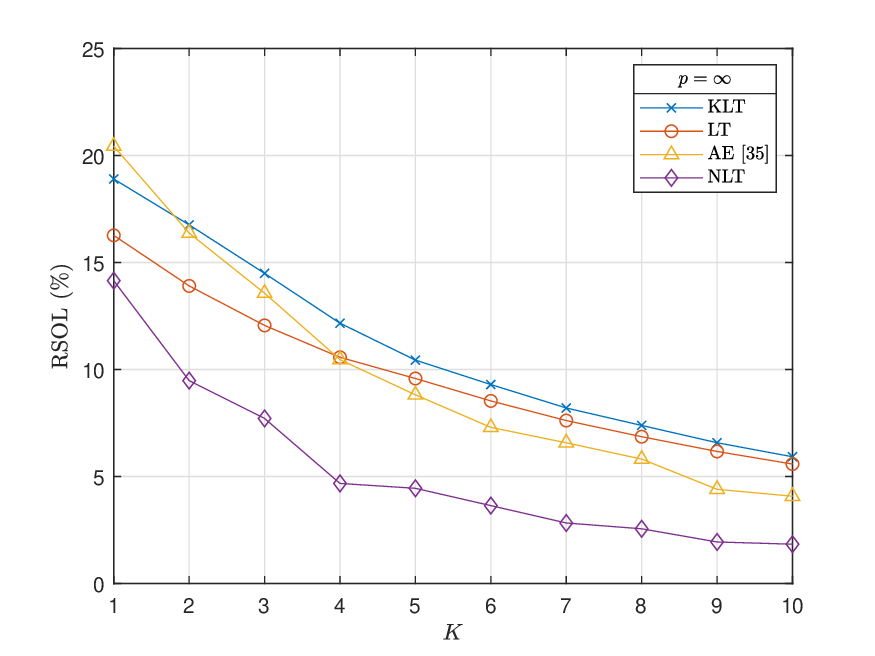}
    \caption{\textcolor{black}{Relative squared optimality loss ($\%$) v.s. the dimension of the precoding output signal space / latent space ($K$).}}
	\label{fig:comparewrtK}
\end{figure}

Additionally, to interpret the obtained gain in a more intuitive manner, we choose one representative realization of $\ell$ from Ausgrid, and display the series of $\ell$, and $x^{\star}(\widehat{\ell})+\ell$ with different transformation methods in Figure~\ref{fig:sample}. By  approximating $\ell$ with a single basis vector, one can observe that the shape of $\widehat{\ell}$ in goal-oriented transformations is more similar to $\ell$ compared with the KLT. As a result, there are fewer fluctuations in $x^{\star}(\widehat{\ell})+\ell$, making it easier to minimize the $L_p$ norm using the proposed LT and NLT schemes.  
\begin{figure}
    \centering

    \begin{subfigure}[b]{0.49\textwidth}
        \centering\includegraphics[scale=0.49]{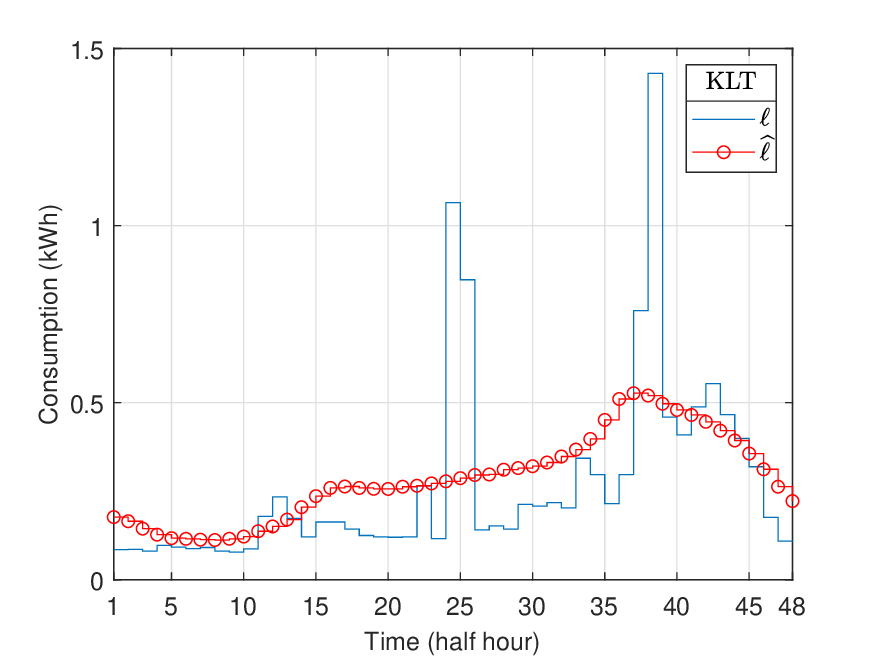}
    \end{subfigure}
    \begin{subfigure}[b]{0.49\textwidth}
        \centering\includegraphics[scale=0.49]{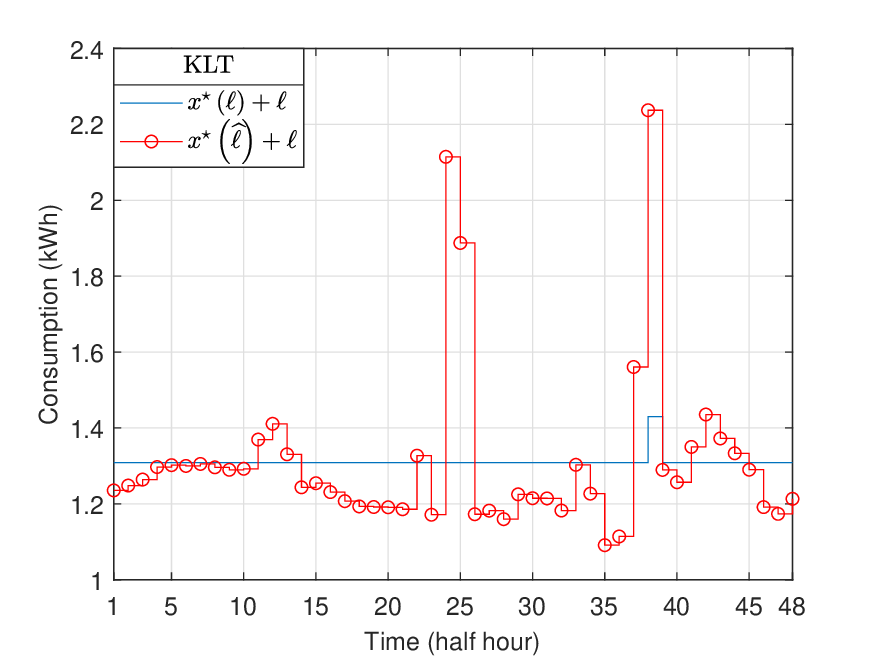}
    \end{subfigure}

    \begin{subfigure}[b]{0.49\textwidth}
        \centering\includegraphics[scale=0.49]{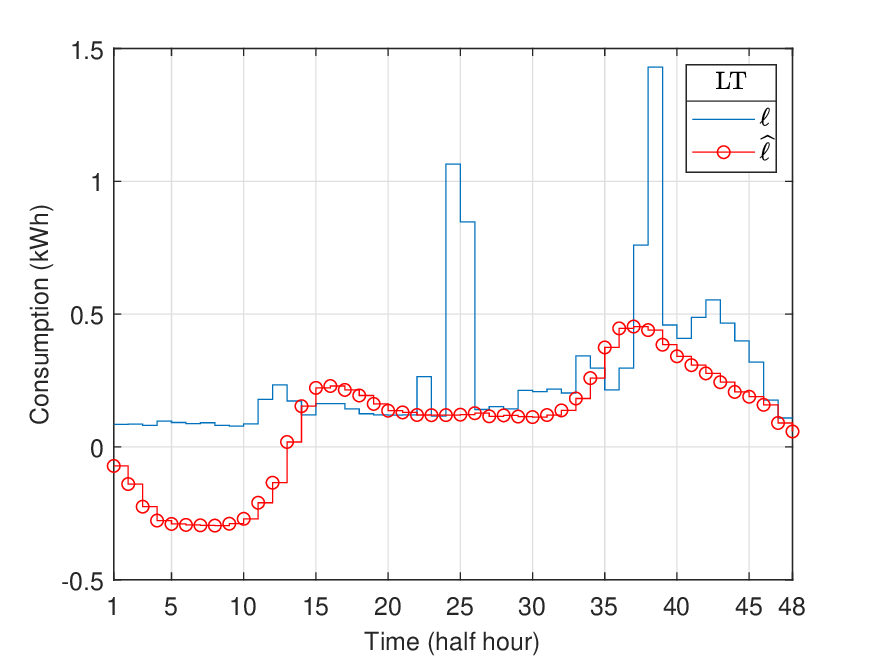}
    \end{subfigure}
    \begin{subfigure}[b]{0.49\textwidth}
        \centering\includegraphics[scale=0.49]{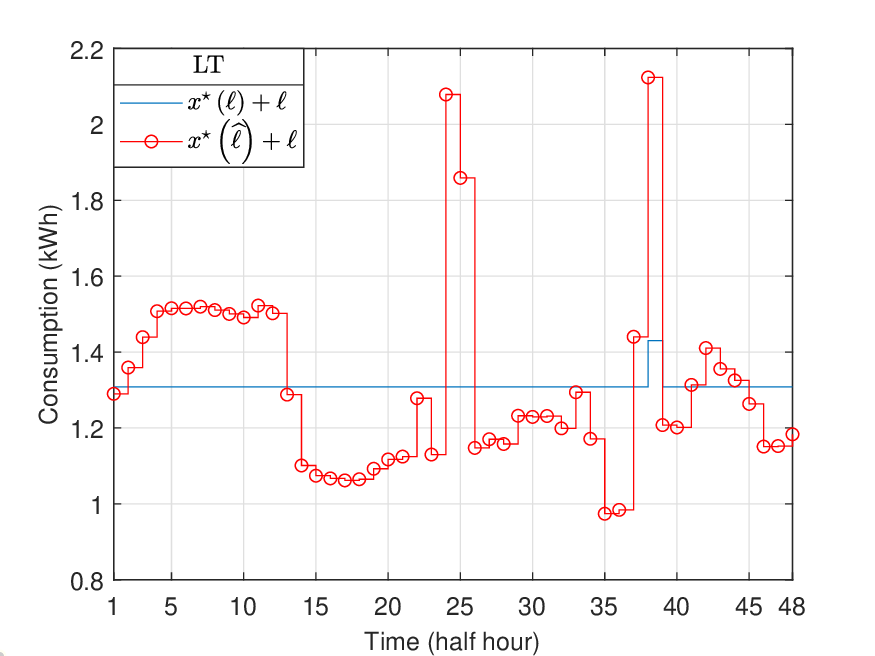}
    \end{subfigure}
    
    \begin{subfigure}[b]{0.49\textwidth}
        \centering\includegraphics[scale=0.49]{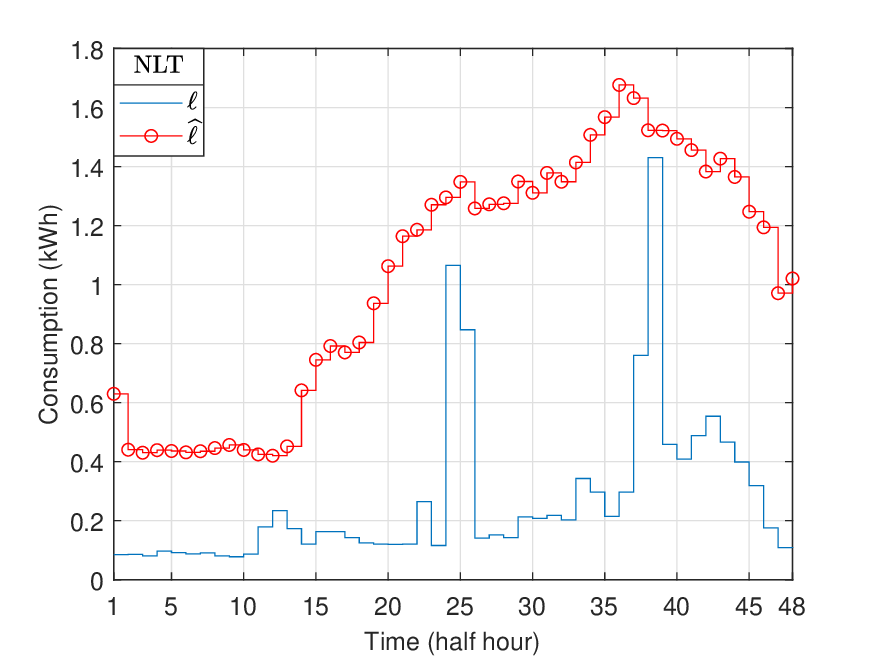}
    \end{subfigure}
    \begin{subfigure}[b]{0.49\textwidth}
        \centering\includegraphics[scale=0.49]{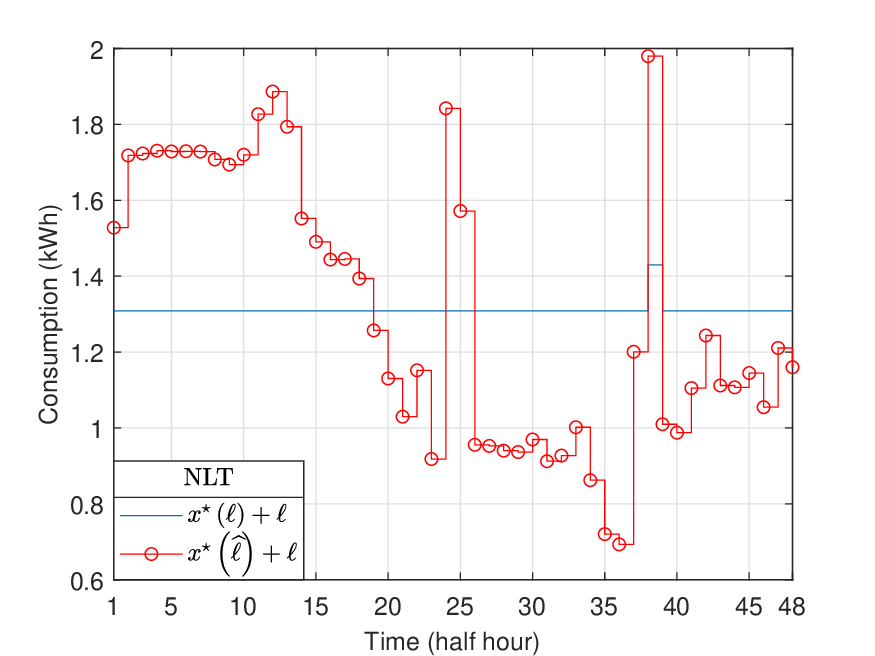}
    \end{subfigure}
    
    \centering
    \caption{Daily energy consumption $\ell$ of one user and its approximations with the corresponding utility entries obtained by three different methods (when $K=1$). NLT has the best performance of minimization of peak value.}
    \label{fig:sample}
\end{figure}

\subsection{Quantization}
After assessing the impact of transformation, we further study the goal-oriented quantization with given transformation schemes. We use the above-mentioned two transformation schemes, followed by a quantizer, and compare the performance with different quantization techniques. When the precoding process reduces the dimension of $\ell$ to $K=1$, we compare three quantization techniques, namely, uniform quantization, LBG and GOQ. Figure~\ref{fig:comparesamplebitnew} represents the RSOL against the number of quantization bits. 
When implementing LT and uniform quantization, we use the hardware-limited task-based quantization (HLTB) algorithm proposed by \citep{shlezinger2019hardware}. Figure~\ref{fig:comparesamplebitnew} shows the GOQ exploits the quantization resources more efficiently such that the optimality loss can be reduced significantly, especially with non-linear transformation. 

Another important issue to discuss here is the choice of the reduced dimension $K$ for {\color{black} a given total bit budget}. Intuitively, when the number of bits tends to infinity, it is not necessary to reduce the dimension in the precoding stage since each element of the transformed vector can be quantized with negligible errors. 
Conversely,  when the total bit budget is very small, a significant reduction in dimensionality is expected, as these bits need to transmit the most important features. Based on these observations, we explore the tradeoff between information loss induced by dimensionality reduction and information loss induced by quantization with finite bit budget. To determine the optimal dimension $K$ that minimizes the optimality loss, we implement the proposed linear transformation scheme and the GOQ method. 
Figure~\ref{fig:cnnautoquantized} illustrates that the RSOL generally first decreases and then increases as $K$ becomes larger, indicating the existence of an optimal dimension for the quantization input. 
\textcolor{black}{This can be explained by the fact that we are operating at a fixed budget on the total number of quantization bits. Therefore, when $K$ increases, the number of bits per dimension decreases, hence the observed phenomenon. This analysis highlights the importance of selecting an appropriate $K$ for dimensionality reduction in the context of quantization, balancing the reduction in information loss due to dimensionality reduction with the information loss introduced by quantization. 
In addition, we address the rate-relative square optimality loss tradeoff in Fig.~\ref{fig:LTGOVQvsBalle} for different values of $K$. For the AE, the compromise is tuned by properly weighting the MSE and the rate in the loss function. The network structure has to be retrained for each $K$ and each value of the weights. In contract with the AE structure, our method can achieve a given level of RSOL with much less transmission bits.}

\begin{figure} \centering\includegraphics[scale=0.7]{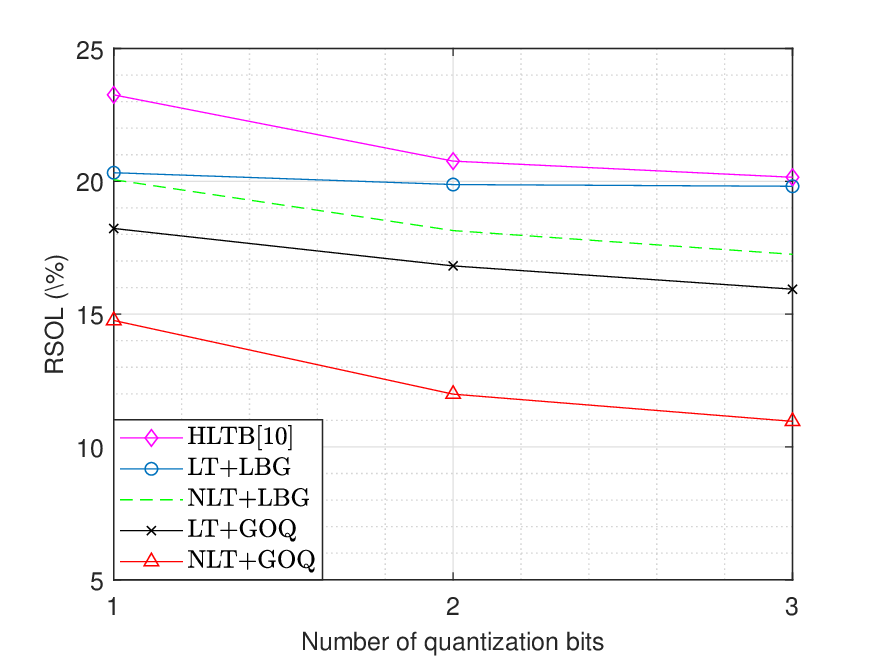}
    \caption{Evolution of the relative optimality loss ($\%$) with quantization bit constraint ($N=48$). RSOL is reduced with larger budget of quantization bits. Among these methods, NLT followed by a GOQ has the best performance.}
	\label{fig:comparesamplebitnew}
\end{figure}

\begin{figure}[H]
\centering
\includegraphics[scale=0.7]{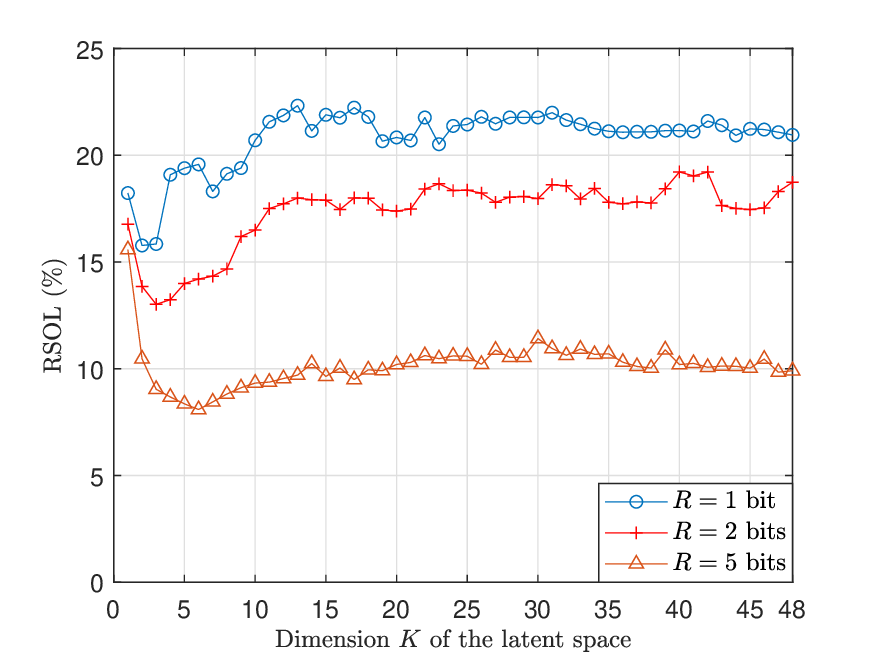}
\caption{Relative optimality loss v.s. the dimension of encoded space ($K$) with a fixed quantization bit constraint for linear transformation and goal-oriented quantization}
\label{fig:cnnautoquantized}
\end{figure}

\begin{figure}[H]
\centering
\includegraphics[scale=0.7]{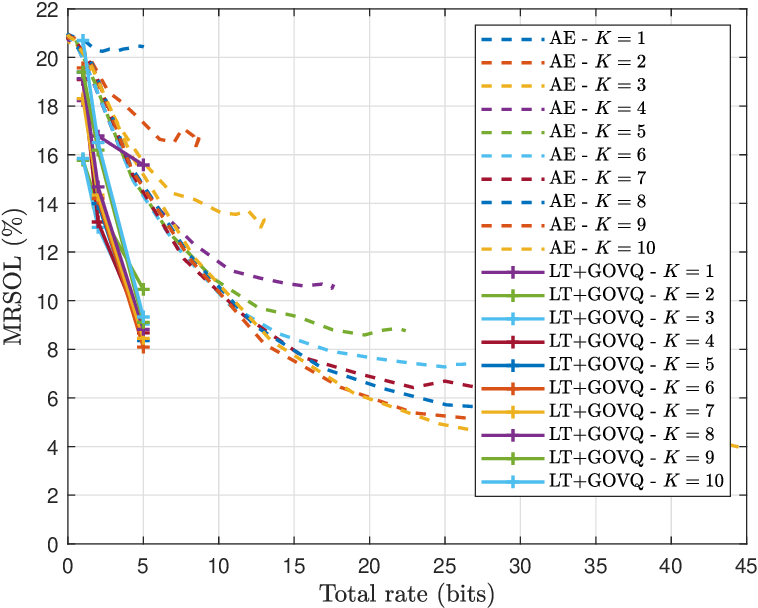}
 \caption{\textcolor{black}{Relative squared optimality loss v.s. the total rate as a function for different latent space dimensions ($K$) for the linear transform followed by the goal-oriented vector quantizer and for the auto-encoder structure adapted from \citep{TensorflowCompTut}.}}
\label{fig:LTGOVQvsBalle}
\end{figure}

\subsection{Iterative algorithm}
At last, we assess the performance of the proposed iterative algorithm by optimizing the transformation scheme and the quantization rule in an alternative way. By taking into account the interplay between quantization noise and the transformation scheme, Figure~\ref{fig:GOPGOQpinf} shows that the iterative algorithm outperforms the aforementioned method using LT and GOQ once.  Although the iterative algorithm provides only a marginal improvement beyond the initial optimization, it demonstrates the potential for further enhancement through careful refinement of the transformation scheme and the quantization rule. \textcolor{black}{Regarding the computation cost,  while iterative optimization has its merits in situations where resources are abundant and  performance maximization is required, a single iteration suffices in resource-limited scenarios, providing a pragmatic balance between performance and computational efficiency. }

\begin{figure}[H]
\centering\includegraphics[scale=0.7]{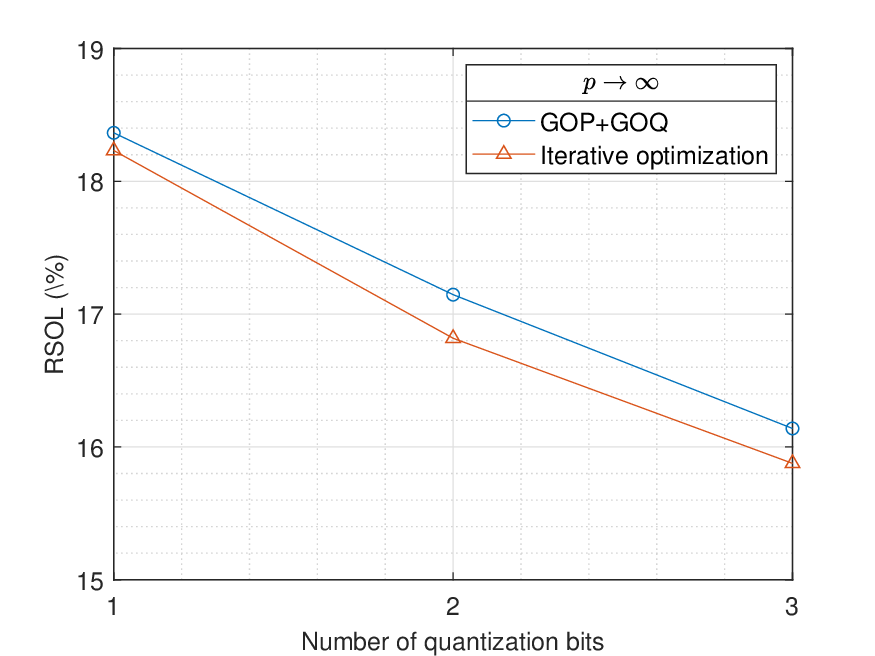}
\caption{Comparing the result of the iterative optimization algorithm to the result without optimizing ($K=1,\ p=\infty$)}
\label{fig:GOPGOQpinf}
\end{figure}
\section{Conclusion}
\label{sec:conclusion}

In this paper, the task of the receiver is modeled by an optimization problem. The goal of the receiver is to minimize an $L_p$-norm performance metric. For taking its decision, the receiver is assumed to have only access to a compressed version of the parameters of the function (minus an $L_p$-norm) to be maximized. The problem of designing a goal-oriented precoder which is followed by a goal-oriented quantizer is addressed. By adopting a non-joint design approach for these two stages and making appropriate approximations, the problem becomes tractable. \textcolor{black}{We provide both an interpretable linear transform which exploits the knowledge of the utility function $u$  and optimal decision, and another non-linear transform based on neural networks.} Compared to the KLT, the proposed linear transformation yields significant performance gains in terms of relative optimality loss. It is also seen to what extent the CNN-based nonlinear transformation performs better than the proposed linear transformation. By moving from the KLT to the CNN-based goal-oriented precoding, the relative optimality loss has been to drop from $20\%$ to values as small as $1-2\%$, which shows all merits of matching the coding scheme to the goal. It is also seen under which conditions, the benefit from using a goal-oriented quantization can accumulate with those from the goal-oriented precoding stage. All these positive results show the interest of adapting, possibly on the fly, the coding scheme to the task of the receiver. The proposed approach might be extended to other settings of practical interest. For instance, when the goal function is not known but only realizations of it are available. This would naturally lead to a reinforcement learning perspective for the design of goal-oriented encoders. Also, in this paper, the communication channel between the encoder and decoder is assumed to be perfect, which leaves space for improvements in the presence of communication noise \textcolor{black}{as considered in \citep{dai2022communication, zhang2023highly}}.

\appendix
\section{Proofs}

\subsection{Proof of Proposition \ref{prop:opdecision}\label{subsec:Proof-opdecison}}

\textcolor{black}{We want to prove that the solution $x^{\star}$ which maximizes (1) under the constraints (2) and (3) is a water-filling solution. For this, we first notice that (1) is a concave function and then apply KKT conditions, which are necessary and sufficient optimality conditions since (2) and (3) are affine constraints.}

\textcolor{black}{Consider a given vector $\ell$ and assume without loss
of generality that} 
\begin{equation}
\ell_{1}\leqslant\ell_{2}\leqslant\dots\leqslant\ell_{j}\leqslant\ell_{j+1}\leqslant\dots\leqslant\ell_{N}.\label{eq:assumption}
\end{equation}
\textcolor{black}{The considered optimization problem \eqref{eq:xstar} is a
convex problem since the Lp-norm is convex and the constraint functions are affine.} Introducing the Lagrangian
\begin{equation}
\mathcal{L}=-\sum_{k=1}^{N}\left(x_{k}+\ell_{k}\right)^{p}-\delta\left(E-\sum_{k=1}^{N}x_{k}\right)-\sum_{k=1}^{N}\lambda_{k}x_{k}
\end{equation}
and \textcolor{black}{applying} KKT conditions, one obtains, as $p\geqslant1$
\begin{equation}
\frac{\partial\mathcal{L}}{\partial x_{k}}=-p\left(x_{k}+\ell_{k}\right)^{p-1}+\delta-\lambda_{k},
\end{equation}
and
\begin{equation}
\begin{cases}
-p\left(x_{k}^{\star}+\ell_{k}\right)^{p-1}+\delta^{\star}-\lambda_{k}^{\star} & =0\\
E-\sum_{i=1}^{N}x_{i}^{\star} & =0\\
-x_{k}^{\star} & \leqslant0\\
\lambda_{k}^{\star}x_{k}^{\star} & =0.
\end{cases}
\end{equation}

Getting rid of the slack variables $\left\{ \lambda_{k}^{\star}\right\} $,
one obtains
\begin{equation}
p\left(x_{k}^{\star}+\ell_{k}\right)^{p-1}=\delta^{\star}
\end{equation}
\begin{equation}
x_{k}^{\star}=\left(\left(\frac{\delta^{\star}}{p}\right)^\frac{1}{p-1}-\ell_{k}\right)^{+}
\end{equation}
where $\left(a\right)^{+}=\max(a,0)$. This is \textcolor{black}{therefore} a water-filling solution
with water level $\mu=\left(\frac{\delta^{\star}}{p}\right)^\frac{1}{p-1}$,
thus
\begin{equation}
x_{j}^{\star}=\left(\mu-\ell_{j}\right)^{+}.\label{eq:waterload_bis}
\end{equation}
Let $n^{\star}$ be the number of non-zero entries of $x^{\star}$.
According to \eqref{eq:assumption}, one has $x_{j}^{\star}>0$ for
$j=1,\dots,n^{\star}$ and $x_{j}^{\star}=0$ for $j>n^{\star}$.
Consequently, as
\begin{equation}
\sum_{j=1}^{n^{*}}x_{j}^{\star}=E,
\end{equation}
one obtains
\begin{equation}
\sum_{j=1}^{n^{*}}\left(\mu-\ell_{j}\right)=E,
\end{equation}
leading to
\begin{equation}
\mu=\frac{1}{n^{\star}}\left(E+\sum_{j=1}^{n^{\star}}\ell_{j}\right).\label{eq:waterlevel}
\end{equation}

The value of $n^{\star}$ is then obtained as the largest value of
$n$ such that $\mu-\ell_{n}>0$
\begin{align*}
n^{\star} & =\arg\max_{n}n\\
 & \text{s.t. }\frac{1}{n}\left(E+\sum_{j=1}^{n}\ell_{j}\right)-\ell_{n}>0.
\end{align*}

\subsection{Proof of Proposition \ref{prop:linearapp}\label{subsec:Proof-linearapp}}

As seen in \ref{subsec:Proof-opdecison}, $x^{\star}\left(\ell\right)=\left[x_{1}^{\star},x_{2}^{\star},\dots,x_{N}^{\star}\right]^{\mathrm{T}}$
is the function of $\ell$. {\color{black}As $n^{\star}\in\{1,\dots,N\}$, $2^{N}-1$ different subsets of $\mathbb{R}_{+}^{N}$
may be defined as}
\[
\mathcal{M}_{\mathcal{I}}=\{\ell\in\mathbb{R}_{+}^{N}\mid x_{k}^{\star}\left(\ell\right)>0,\ k\in\mathcal{I},\ x_{j}^{\star}\left(\ell\right)=0,\ j\notin\mathcal{I}\}
\]
where 
\[
\mathcal{I}\subseteq\{1,2,\dots,N\},\ \mathcal{I}\neq\varnothing
\]
is the set of indexes of the non-zero entries of $x^{\star}$.

In each region $\mathcal{M}_{\mathcal{I}}$, a linear expression
of the solution of the optimization problem \eqref{eq:xstar} can
be obtained. Consider some $\mathcal{I}$ and
the associated region $\mathcal{M}_{\mathcal{I}}$. We assume without
loss of generality that the last $N-n^{\star}$ elements of $x^{\star}$
are null for all $\ell\in\mathcal{M}_{\mathcal{I}}$. From
\eqref{eq:waterload_bis} and \eqref{eq:waterlevel}, one has 
\begin{equation}
x_{j}^{\star}\left(\ell\right)=\begin{cases}
\frac{\sum_{i=1}^{n^{\star}}\ell_{i}+E}{n^{\star}}-\ell_{j} & j\leqslant n^{\star}\\
0 & j>n^{\star}
\end{cases}\label{eq:xstar_l}
\end{equation}
the Jacobian matrix of $x^{\star}\left(\ell\right)$ is 
\begin{equation}
\begin{aligned} & \boldsymbol{H}\left(\ell\right)=\\
 & \left(\begin{array}{ccccccc}
-1+\frac{1}{n^{\star}} & \frac{1}{n^{\star}} & \text{\ensuremath{\cdots}} & \frac{1}{n^{\star}} & 0 & \cdots & 0\\
\frac{1}{n^{\star}} & -1+\frac{1}{n^{\star}} & \cdots & \vdots & \vdots &  & \vdots\\
\vdots & \vdots &  & \vdots & \vdots &  & \vdots\\
\frac{1}{n^{\star}} & \cdots & \cdots & -1+\frac{1}{n^{\star}} & 0 & \cdots & 0\\
0 & \cdots & \cdots & 0 & 0 & \cdots & 0\\
\vdots &  &  & \vdots & \vdots &  & \vdots\\
0 & \cdots & \cdots & 0 & 0 & \cdots & 0
\end{array}\right).
\end{aligned}
\end{equation}
Using \eqref{eq:xstar_l} and the expression of $\boldsymbol{H}\left(\ell\right)$,
one obtains
\[
x^{\star}\left(\ell\right)=\boldsymbol{H}\left(\ell\right)\ell+b\left(\ell\right).
\]

\subsection{Proof of Proposition \ref{prop:gradient}\label{subsec:Proof-gradient}}

From \eqref{eq:GammaB}, when $p\in\mathbb{N}^{+}$ 
\[
\begin{aligned}\frac{\partial\widehat{\Gamma}_T}{\partial\boldsymbol{B}}\left(\boldsymbol{B}\right)=- & \frac{1}{T}\sum_{i=1}^{T}2\left(u\left(x^{\star}\left(\ell^{\left(i\right)}\right);\ell^{\left(i\right)}\right)-u\left(x^{\star}\left(\widehat{\ell}^{\left(i\right)}\right);\ell^{\left(i\right)}\right)\right)\frac{\partial\mathcal{G}_{i}}{\partial\boldsymbol{B}}\end{aligned}
\]
where 
\[
\begin{aligned}\mathcal{G}_{i} & =u\left(x^{\star}\left(\widehat{\ell}^{\left(i\right)}\right);\ell^{\left(i\right)}\right)\\
 & =\left(\mathbf{1}^{\mathrm{T}}\cdot\underbrace{\left(x^{\star}\left(\widehat{\ell}^{\left(i\right)}\right)+\ell^{\left(i\right)}\right)\odot\cdots\odot\left(x^{\star}\left(\widehat{\ell}^{\left(i\right)}\right)+\ell^{\left(i\right)}\right)}_{p}\right)^{\frac{1}{p}}
\end{aligned}
\]
with $\mathbf{1}=\left[1,\dots,1\right]^{\mathrm{T}}$ the vector
of $N$ ones. Then
\[
\begin{aligned}d\mathcal{G}_{i}= & \frac{1}{p}\|x^{\star}\left(\widehat{\ell}^{\left(i\right)}\right)+\ell^{\left(i\right)}\|_{p}^{1-p}*\\
 & d\left(\mathbf{1}^{\mathrm{T}}\underbrace{\left(x^{\star}\left(\widehat{\ell}^{\left(i\right)}\right)+\ell^{\left(i\right)}\right)\odot\cdots\odot\left(x^{\star}\left(\widehat{\ell}^{\left(i\right)}\right)+\ell^{\left(i\right)}\right)}_{p}\right)\\
= & \frac{1}{p}\|x^{\star}\left(\widehat{\ell}^{\left(i\right)}\right)+\ell^{\left(i\right)}\|_{p}^{1-p}*p*\beta_{i}^{\mathrm{T}}d\left(x^{\star}\left(\widehat{\ell}^{\left(i\right)}\right)+\ell^{\left(i\right)}\right)\\
= & \|x^{\star}\left(\widehat{\ell}^{\left(i\right)}\right)+\ell^{\left(i\right)}\|_{p}^{1-p}\beta_{i}^{\mathrm{T}}d\left(\left(\boldsymbol{H}_{i}\boldsymbol{B}^{\mathrm{T}}\boldsymbol{B}+\boldsymbol{I}\right)\ell^{\left(i\right)}+b_{i}\right)\\
= & \|x^{\star}\left(\widehat{\ell}^{\left(i\right)}\right)+\ell^{\left(i\right)}\|_{p}^{1-p}\beta_{i}^{\mathrm{T}}\left(\boldsymbol{H}_{i}d\boldsymbol{B}^{\mathrm{T}}\boldsymbol{B}\ell^{\left(i\right)}+\boldsymbol{H}_{i}\boldsymbol{B}^{\mathrm{T}}d\boldsymbol{B}\ell^{\left(i\right)}\right)
\end{aligned}
\]
where $\boldsymbol{H}_{i}=\boldsymbol{H}\left(\widehat{\ell}^{\left(i\right)}\right)$,
$b_{i}=b\left(\widehat{\ell}^{\left(i\right)}\right)$, and 
\[
\beta_{i}=\text{\ensuremath{\underbrace{\left(x^{\star}\left(\widehat{\ell}^{\left(i\right)}\right)+\ell^{\left(i\right)}\right)\odot\cdots\odot\left(x^{\star}\left(\widehat{\ell}^{\left(i\right)}\right)+\ell^{\left(i\right)}\right)}_{p-1}}}.
\]

As $d\mathcal{G}_{i}$ is a scalar, one has 
\[
\begin{aligned}d\mathcal{G}_{i}= & \text{Tr}\left(d\mathcal{G}_{i}\right)\\
= & \|x^{\star}\left(\widehat{\ell}^{\left(i\right)}\right)+\ell^{\left(i\right)}\|_{p}^{1-p} \left(\text{Tr}\left(\beta_{i}^{\mathrm{T}}\boldsymbol{H}_{i}d\boldsymbol{B}^{\mathrm{T}}\boldsymbol{B}\ell^{\left(i\right)}\right)+\text{Tr}\left(\beta_{i}^{\mathrm{T}}\boldsymbol{H}_{i}\boldsymbol{B}^{\mathrm{T}}d\boldsymbol{B}\ell^{\left(i\right)}\right)\right)\\
= & \|x^{\star}\left(\widehat{\ell}^{\left(i\right)}\right)+\ell^{\left(i\right)}\|_{p}^{1-p} \left(\text{Tr}\left(\boldsymbol{B}\ell^{\left(i\right)}\beta_{i}^{\mathrm{T}}\boldsymbol{H}_{i}d\boldsymbol{B}^{\mathrm{T}}\right)+\text{Tr}\left(\ell^{\left(i\right)}\beta_{i}^{\mathrm{T}}\boldsymbol{H}_{i}\boldsymbol{B}^{\mathrm{T}}d\boldsymbol{B}\right)\right)\\
= & \|x^{\star}\left(\widehat{\ell}^{\left(i\right)}\right)+\ell^{\left(i\right)}\|_{p}^{1-p} \text{Tr}\left(\left(\boldsymbol{B}\ell^{\left(i\right)}\beta_{i}^{\mathrm{T}}\boldsymbol{H}_{i}+\boldsymbol{B}\boldsymbol{H}_{i}^{\mathrm{T}}\beta_{i}\ell^{\left(i\right)T}\right)^{\mathrm{T}}d\boldsymbol{B}\right).
\end{aligned}
\]
Then, one can deduce 
\[
\frac{\partial\mathcal{G}_{i}}{\partial\boldsymbol{B}}=\|x^{\star}\left(\widehat{\ell}^{\left(i\right)}\right)+\ell^{\left(i\right)}\|_{p}^{1-p}\left(\boldsymbol{B}\ell^{\left(i\right)}\beta_{i}^{\mathrm{T}}\boldsymbol{H}_{i}+\boldsymbol{B}\boldsymbol{H}_{i}^{\mathrm{T}}\beta_{i}\ell^{\left(i\right)T}\right)
\]
and
\begin{equation}
\frac{\partial\widehat{\Gamma}_T}{\partial\boldsymbol{B}}\left(\boldsymbol{B}\right)=-\frac{1}{T}\sum_{i=1}^{\mathrm{T}}C_{i}\left(\boldsymbol{B}\ell^{\left(i\right)}\beta_{i}^{\mathrm{T}}\boldsymbol{H}_{i}+\boldsymbol{B}\boldsymbol{H}_{i}^{\mathrm{T}}\beta_{i}\ell^{\left(i\right)T}\right)
\end{equation}
where 
\begin{equation}
\begin{aligned}C_{i}=2 & \left(u\left(x^{\star}\left(\ell^{\left(i\right)}\right);\ell^{\left(i\right)}\right)-u\left(x^{\star}\left(\widehat{\ell}^{\left(i\right)}\right);\ell^{\left(i\right)}\right)\right)||x^{\star}\left(\widehat{\ell}^{\left(i\right)}\right)+\ell^{\left(i\right)}||_{p}^{1-p}.\end{aligned}
\end{equation}

When $p\rightarrow+\infty$, from \eqref{eq:GammaB}, we have 
\[
\begin{aligned}\frac{\partial\widehat{\Gamma}_T}{\partial\boldsymbol{B}}\left(\boldsymbol{B}\right) = & -\frac{1}{T}\sum_{i=1}^{T}2\left(u\left(x^{\star}\left(\ell^{\left(i\right)}\right);\ell^{\left(i\right)}\right)-u\left(x^{\star}\left(\widehat{\ell}^{\left(i\right)}\right);\ell^{\left(i\right)}\right)\right)\frac{\partial\mathcal{F}_{i}}{\partial\boldsymbol{B}}\end{aligned}
\]
with 
\[
\mathcal{F}_{i}=\max_{j}\left(\widehat{x}_{j}^{\left(i\right)}+\ell_{j}^{\left(i\right)}\right)
\]
and $x^{\star}(\widehat{\ell}^{\left(i\right)})=\left[\widehat{x}_{1}^{\left(i\right)},\widehat{x}_{2}^{\left(i\right)},\cdots,\widehat{x}_{N}^{\left(i\right)}\right]^{\mathrm{T}}$.

Furthermore, we assume for $i$-th vector, the $k\left(i\right)$-th
elements is maximum. Then 
\[
\begin{aligned}\mathcal{F_{\textmd{\textit{i}}}}= & s_{k\left(i\right)}^{\mathrm{T}}\left(x\left(\widehat{\ell}^{\left(i\right)}\right)+\ell^{\left(i\right)}\right)\\
= & s_{k\left(i\right)}^{\mathrm{T}}\left(\left(\boldsymbol{H}_{i}\boldsymbol{B}^{\mathrm{T}}\boldsymbol{B}+\boldsymbol{I}\right)\ell^{\left(i\right)}+b_{i}\right)
\end{aligned}
\]
with 
\[
s_{k\left(i\right)}=\left(\begin{array}{c}
0_{(k\left(i\right)-1)\times1}\\
1\\
0_{(N-k\left(i\right))\times1}
\end{array}\right).
\]

We have 
\[
\begin{aligned}d\mathcal{F}_{\textmd{\textit{i}}} & =s_{k\left(i\right)}^{\mathrm{T}}\left(\boldsymbol{H}_{i}d\boldsymbol{B}^{\mathrm{T}}\boldsymbol{B}\ell^{\left(i\right)}+\boldsymbol{H}_{i}\boldsymbol{B}^{\mathrm{T}}d\boldsymbol{B}\ell^{\left(i\right)}\right).\end{aligned}
\]
\[
\begin{aligned}\text{Tr}\left(d\mathcal{F}_{\textmd{\textit{i}}}\right) & =\text{Tr}\left(\left(\boldsymbol{B}\ell^{\left(i\right)}s_{k\left(i\right)}^{\mathrm{T}}\boldsymbol{H}_{i}+\boldsymbol{B}\boldsymbol{H}_{i}^{\mathrm{T}}s_{k\left(i\right)}\ell^{\left(i\right)T}\right)^{\mathrm{T}}d\boldsymbol{B}\right)\end{aligned}
.
\]
Consequently,
\[
\frac{\partial\mathcal{F}_{i}}{\partial\boldsymbol{B}}=\boldsymbol{B}\ell^{\left(i\right)}s_{k\left(i\right)}^{\mathrm{T}}\boldsymbol{H}_{i}+\boldsymbol{B}\boldsymbol{H}_{i}^{\mathrm{T}}s_{k\left(i\right)}\ell^{\left(i\right)T}
\]
and
\begin{equation}
\frac{\partial\widehat{\Gamma}_T}{\partial\boldsymbol{B}}\left(\boldsymbol{B}\right)=-\frac{1}{T}\sum_{i=1}^{T}D_{i}\left(\boldsymbol{B}\ell^{\left(i\right)}s_{k\left(i\right)}^{\mathrm{T}}\boldsymbol{H}_{i}+\boldsymbol{B}\boldsymbol{H}_{i}^{\mathrm{T}}s_{k\left(i\right)}\ell^{\left(i\right)T}\right)
\end{equation}
with 
\begin{equation}
D_{i}=2\left(u\left(x^{\star}\left(\ell^{\left(i\right)}\right);\ell^{\left(i\right)}\right)-u\left(x^{\star}\left(\widehat{\ell}^{\left(i\right)}\right);\ell^{\left(i\right)}\right)\right).
\end{equation}




\bibliographystyle{elsarticle-num} 
\bibliography{IEEETran}





\end{document}